\documentclass[a4paper,USenglish]{lipics-v2018-adapted}

\usepackage{setspace}
\usepackage{amsthm, amsmath}
\usepackage{xcolor}
\usepackage[linesnumbered,algoruled,vlined]{algorithm2e}
\usepackage{xparse} 
\usepackage{array,multirow}
\usepackage{graphicx}
\usepackage{enumitem}
\usepackage{comment}
\usepackage{stmaryrd}
\usepackage{url}
\usepackage{ifthen}
\usepackage{wrapfig}

\newcommand{\CompilingOnStefansSystem}{0}

\ifthenelse{\equal{\CompilingOnStefansSystem}{1}}{
	\DeclareDocumentCommand{\restrict}{d[] d[]}{\ensuremath{\langle #1 \vert #2 \rangle}}
	\DeclareDocumentCommand{\subLP}{d[] d[]}{\ensuremath{#1 [| #2 |]}}
}{
	\makeatletter
	\DeclareFontEncoding{LS1}{}{\noaccents@}
	\DeclareFontEncoding{LS2}{}{\@noaccents}
	\makeatother
	\DeclareFontSubstitution{LS1}{stix}{m}{n}
	\DeclareFontSubstitution{LS2}{stix}{m}{n}

	\DeclareSymbolFont{stixSymbolFont}{LS1}{stixfrak} {m} {n}
	\DeclareSymbolFont{largesymbolsstix}{LS2}{stixex}{m}{n}

	\DeclareMathDelimiter{\llangle}{\mathopen}{stixSymbolFont}{"28}{stixSymbolFont}{"28}
	\DeclareMathDelimiter{\rrangle}{\mathclose}{stixSymbolFont}{"29}{stixSymbolFont}{"29}

	\DeclareMathDelimiter{\lbrbrak}{\mathopen}{largesymbolsstix}{"EE}{largesymbolsstix}{"14}
	\DeclareMathDelimiter{\rbrbrak}{\mathclose}{largesymbolsstix}{"EF}{largesymbolsstix}{"15}

	\DeclareDocumentCommand{\restrict}{d[] d[]}{\ensuremath{\llangle #1 \vert #2 \rrangle}}
	\DeclareDocumentCommand{\subLP}{d[] d[]}{\ensuremath{#1 \llbracket #2 \rrbracket}}

}

  \newcolumntype{F}{>{$\displaystyle}r<{$}@{\hspace{0.0em}}}
  \newcolumntype{C}{>{$\displaystyle\,}c<{$}@{\hspace{0.0em}}}
  \newcolumntype{B}{>{$\displaystyle\,}r<{$}@{\hspace{0.0em}}}
  \newcolumntype{R}{>{$\displaystyle}r<{$}@{\hspace{0.2em}}}
  \newcolumntype{S}{>{$\displaystyle}r<{$}@{\hspace{0.2em}}}
  \newcolumntype{L}{>{$\displaystyle}l<{$}@{\hspace{0.2em}}}
  \newcolumntype{Q}{>{$\displaystyle}l<{$}@{\hspace{0.3em}}}

 \newcounter{IPnumber}
 \newcommand{\tagIt}[1]{\refstepcounter{equation}\textnormal{({\theequation})} \label{#1}}
 
 \makeatletter
 
 \newcommand{\eqnum}{\leavevmode\hfill\refstepcounter{equation}\textup{\tagform@{\theequation}}}
 \makeatother

\usepackage{amsthm,amssymb}
\usepackage{thmtools,thm-restate}

 \theoremstyle{plain}

\declaretheorem[style=definition,qed=$\square$,sibling=theorem]{definition}

\newtheorem{observation}[theorem]{Observation}

\usepackage{microtype}

\title{(FPT-)Approximation Algorithms for the \newline Virtual Network Embedding Problem}

\titlerunning{(FPT-)Approximation Algorithms for the Virtual Network Embedding Problem}

\author{Matthias Rost}{TU Berlin, Germany{}}{}{}{}

\author{Stefan Schmid}{University of Vienna, Austria{}}{}{}{}

\authorrunning{M. Rost and S. Schmid}

\Copyright{Matthias Rost and Stefan Schmid}

\subjclass{}


\category{}

\relatedversion{}

\supplement{}

\funding{}

\acknowledgements{}

\nolinenumbers 
\hideLIPIcs  

\begin{document}

\maketitle

\newcommand{\TODO}[1]{\textcolor{red}{TODO: #1}}

\newcommand\numberthis{\addtocounter{equation}{1}\tag{\theequation}}

\newcommand{\NULL}{\textnormal{\texttt{NULL}}}
\newcommand{\scale}{\ensuremath{\lambda}}


\newcommand{\preals}{\ensuremath{\mathbb{R}_{\geq 0}}}


\newcommand{\requests}{\ensuremath{\mathcal{R}}}
\newcommand{\requestsP}{\ensuremath{\mathcal{R}'}}
\newcommand{\req}{r}
\newcommand{\types}{\ensuremath{\mathcal{T}}}
\newcommand{\type}{\ensuremath{\tau}}
\newcommand{\SVTypes}[1][\type]{\ensuremath{V^{#1}_{S}}}
\newcommand{\SVTypesCycle}[1][C_k]{\ensuremath{V^{#1}_{S,t}}}

\newcommand{\VG}[1][\req]{\ensuremath{G_{#1}}}
\newcommand{\VV}[1][\req]{\ensuremath{V_{#1}}}
\newcommand{\VE}[1][\req]{\ensuremath{E_{#1}}}
\newcommand{\VGbar}[1][\req]{\ensuremath{\bar{G}_{#1}}}
\newcommand{\VVbar}[1][\req]{\ensuremath{\bar{V}_{#1}}}
\newcommand{\VEbar}[1][\req]{\ensuremath{\bar{E}_{#1}}}
\newcommand{\Vstart}[1][\req]{\ensuremath{s_{#1}}}
\newcommand{\Vend}[1][\req]{\ensuremath{t_{#1}}}

\newcommand{\VGext}[1][\req]{\ensuremath{G^{\textnormal{ext}}_{#1}}}
\newcommand{\VGextFlow}[1][\req]{\ensuremath{G^{\textnormal{ext}}_{#1,f}}}
\NewDocumentCommand{\VGP}{O{\req} O{i} O{j}}{\ensuremath{G^{#2,#3}_{#1}}}
\newcommand{\VVext}[1][\req]{\ensuremath{V^{\textnormal{ext}}_{#1}}}
\newcommand{\VEext}[1][\req]{\ensuremath{E^{\textnormal{ext}}_{#1}}}
\newcommand{\VEextHorizontal}[1][\req]{\ensuremath{E^{\textnormal{ext}}_{#1,u,v}}}
\newcommand{\VEextVertical}[1][\req]{\ensuremath{E^{\textnormal{ext}}_{#1,\type,u}}}

\newcommand{\Vsource}[1][\req]{\ensuremath{o^+_{\req}}}
\newcommand{\Vsink}[1][\req]{\ensuremath{o^-_{\req}}}

\newcommand{\VMultiplicity}[1][\req]{\ensuremath{M_{#1}}}
\newcommand{\Vprofit}[1][\req]{\ensuremath{b_{#1}}}
\newcommand{\VprofitMax}{\ensuremath{b_{\max}}}

\newcommand{\Vcap}[1][\req]{\ensuremath{d_{#1}}}
\newcommand{\VVloc}[1][i]{\ensuremath{V^{{\req,#1}}_{S}}}
\newcommand{\VEloc}[1][i,j]{\ensuremath{E^{{\req,#1}}_{S}}}
\newcommand{\Vtype}[1][\req]{\ensuremath{\tau_{#1}}}


\newcommand{\SG}{\ensuremath{G_S}}
\newcommand{\SR}{\ensuremath{R_{S}}}
\newcommand{\SRV}{\ensuremath{R^V_{S}}}
\newcommand{\SV}{\ensuremath{V_S}}
\newcommand{\SE}{\ensuremath{E_S}}

\newcommand{\Scap}{\ensuremath{d_{S}}}
\newcommand{\ScapType}[1][\type]{\ensuremath{d^{#1}_{\SV}}}
\newcommand{\ScapTypePrime}[1][\type]{\ensuremath{{{d'}^{#1}_{\SV}}}}

\newcommand{\Scost}{\ensuremath{c_{S}}}
\newcommand{\ScostType}[1][\type]{\ensuremath{c^{#1}_{\SV}}}


\newcommand{\map}[1][\req]{\ensuremath{m_{#1}}}
\newcommand{\mapV}[1][\req]{\ensuremath{m^V_{#1}}}
\newcommand{\mapE}[1][\req]{\ensuremath{m^E_{#1}}}

\newcommand{\FeasibleLP}{\ensuremath{\mathcal{F}^{\textnormal{new}}_{\textnormal{LP}}}}
\newcommand{\FeasibleIP}{\ensuremath{\mathcal{F}_{\textnormal{IP}}}}


\makeatletter
\newcommand{\nosemic}{\renewcommand{\@endalgocfline}{\relax}}
\newcommand{\dosemic}{\renewcommand{\@endalgocfline}{\algocf@endline}}
\newcommand{\pushline}{\Indp}
\newcommand{\popline}{\Indm\dosemic}
\let\oldnl\nl
\newcommand{\nonl}{\renewcommand{\nl}{\let\nl\oldnl}}
\makeatother

\makeatletter
\newcommand{\removelatexerror}{\let\@latex@error\@gobble}
\makeatother

\newcounter{ipCounter}
\NewDocumentEnvironment{IPFormulation}{m}{%
\refstepcounter{ipCounter}
\begin{algorithm}[#1]%
\renewcommand\thealgocf{\arabic{ipCounter}}
}{%
\end{algorithm}
\addtocounter{algocf}{-1}
}

\NewDocumentEnvironment{IPFormulationStar}{m}{%
\refstepcounter{ipCounter}
\begin{algorithm*}[#1]%
\renewcommand\thealgocf{\arabic{ipCounter}}
}{%
\end{algorithm*}
\addtocounter{algocf}{-1}
}



\newcommand{\spaceSolReq}[1][\req]{\ensuremath{\mathcal{M}_{#1}}}
\newcommand{\spaceLP}{\ensuremath{\mathcal{F}^{\textnormal{mcf}}_{\textnormal{LP}}}}
\newcommand{\spaceIP}{\ensuremath{\mathcal{F}^{\textnormal{mcf}}_{\textnormal{IP}}}}
\newcommand{\spaceIPCC}{\ensuremath{\mathcal{F}^{\textnormal{}}_{\textnormal{IP}}}}
\newcommand{\spaceLPCC}{\ensuremath{\mathcal{F}^{\textnormal{}}_{\textnormal{LP}}}}
\newcommand{\spaceIPMDK}{\ensuremath{\mathcal{F}^{\textnormal{MDK}}_{\textnormal{IP}}}}
\newcommand{\spaceLPMDK}{\ensuremath{\mathcal{F}^{\textnormal{MDK}}_{\textnormal{LP}}}}
\newcommand{\spaceLPNew}{\ensuremath{\mathcal{F}^{\textnormal{new}}_{\textnormal{LP}}}}
\newcommand{\spaceLPD}{\ensuremath{\mathcal{F}^{\mathcal{D}}_{\textnormal{LP}}}}
\newcommand{\spaceLPDreq}[1][\req]{\ensuremath{\mathcal{F}^{\mathcal{D}}_{\textnormal{LP},#1}}}

\DeclareDocumentCommand{\OptProfit}{}{\ensuremath{\hat{P}}}


\DeclareDocumentCommand{\NodeG}{O{\req} O{\pi}}{\ensuremath{G^N_{#1,#2}}}
\DeclareDocumentCommand{\NodeV}{O{\req} O{\pi}}{\ensuremath{V^N_{#1,#2}}}
\DeclareDocumentCommand{\NodeE}{O{\req} O{\pi}}{\ensuremath{E^N_{#1,#2}}}

\DeclareDocumentCommand{\EdgeG}{O{\req} O{i} O{j} O{u}}{\ensuremath{G^E_{#1,#2,#3,#4}}}
\DeclareDocumentCommand{\EdgeV}{O{\req} O{i} O{j} O{u}}{\ensuremath{V^E_{#1,#2,#3,#4}}}
\DeclareDocumentCommand{\EdgeE}{O{\req} O{i} O{j} O{u}}{\ensuremath{E^E_{#1,#2,#3,#4}}}

\DeclareDocumentCommand{\VESD}{O{\req}}{\ensuremath{\overrightarrow{E}_{#1}}}
\DeclareDocumentCommand{\VEOD}{O{\req}}{\ensuremath{\overleftarrow{E}_{#1}}}
\DeclareDocumentCommand{\VESigmaD}{O{\req} O{\sigma}}{\ensuremath{E_{#1,#2}}}

\DeclareDocumentCommand{\NodeVRange}{O{\req} O{i} O{j}}{\ensuremath{V^N_{#1,#2,#3}}}
\DeclareDocumentCommand{\NodeVRangeRange}{O{\req} O{i} O{j} O{u} O{v}}{\ensuremath{V^N_{#1,#2,#3,#4,#5}}}

\DeclareDocumentCommand{\path}{O{k }}{\ensuremath{P_{#1}}}
\DeclareDocumentCommand{\cycle}{O{k }}{\ensuremath{C_{#1}}}

\DeclareDocumentCommand{\NodePaths}{O{\req}}{\ensuremath{\mathcal{P}_{#1}}}


\DeclareDocumentCommand{\loadV}{O{\req} O{u}}{\ensuremath{a_{#1,#2}}}
\DeclareDocumentCommand{\loadE}{O{\req} O{u} O{v}}{\ensuremath{a_{#1,#2,#3}}}
\DeclareDocumentCommand{\loadX}{O{\req} O{x}}{\ensuremath{a_{#1,#2}}}
\DeclareDocumentCommand{\decomp}{O{\req} O{k}}{\ensuremath{D_{#1}^{#2}}}
\DeclareDocumentCommand{\decompHat}{O{\req} O{k}}{\ensuremath{{\hat{D}}_{#1}^{#2}}}
\DeclareDocumentCommand{\load}{O{\req} O{k}}{\ensuremath{a_{#1}^{#2}}}
\DeclareDocumentCommand{\prob}{O{\req} O{k}}{\ensuremath{f_{#1}^{#2}}}
\DeclareDocumentCommand{\mapping}{O{\req} O{k}}{\ensuremath{m_{#1}^{#2}}}

\DeclareDocumentCommand{\loadHat}{O{\req} O{k}}{\ensuremath{\hat{a}_{#1}^{#2}}}
\DeclareDocumentCommand{\probHat}{O{\req} O{k}}{\ensuremath{\hat{f}_{#1}^{#2}}}
\DeclareDocumentCommand{\mappingHat}{O{\req} O{k}}{\ensuremath{\hat{m}_{#1}^{#2}}}

\DeclareDocumentCommand{\loadHat}{O{\req} O{k}}{\ensuremath{\hat{a}_{#1}^{#2}}}
\DeclareDocumentCommand{\probHat}{O{\req} O{k}}{\ensuremath{\hat{f}_{#1}^{#2}}}
\DeclareDocumentCommand{\mappingHat}{O{\req} O{k}}{\ensuremath{\hat{m}_{#1}^{#2}}}

\DeclareDocumentCommand{\Exp}{}{\ensuremath{\mathbb{E}}}
\DeclareDocumentCommand{\randVarX}{O{\req} O{k}}{\ensuremath{X_{#1}^{#2}}}
\DeclareDocumentCommand{\randVarY}{O{\req}}{\ensuremath{Y_{#1}}}
\DeclareDocumentCommand{\randVarZ}{O{\req}}{\ensuremath{Z_{#1}}}
\DeclareDocumentCommand{\randVarL}{O{x}}{\ensuremath{L_{x}}}
\DeclareDocumentCommand{\randVarLX}{O{\req} O{x} O{y}}{\ensuremath{L_{#1,#2,#3}}}
\DeclareDocumentCommand{\randVarLNode}{O{\req} O{\type} O{u}}{\ensuremath{L_{#1,#2,#3}}}
\DeclareDocumentCommand{\randVarLEdge}{O{\req} O{u} O{v}}{\ensuremath{L_{#1,#2,#3}}}
\DeclareDocumentCommand{\randVarM}{O{\req}}{\ensuremath{M_{#1}}}
\DeclareDocumentCommand{\randVarC}{O{\req}}{\ensuremath{C_{#1}}}

\DeclareDocumentCommand{\ProbVarX}{O{1}}{\ensuremath{\mathbb{P}(\randVarX = #1)}}
\DeclareDocumentCommand{\ProbVarY}{O{1}}{\ensuremath{\mathbb{P}(\randVarY = #1)}}
\DeclareDocumentCommand{\ProbVarZ}{O{1}}{\ensuremath{\mathbb{P}(\randVarZ = #1)}}
\DeclareDocumentCommand{\ProbVarL}{O{1}}{\ensuremath{\mathbb{P}(\randVarL = #1)}}
\DeclareDocumentCommand{\ProbVarM}{O{1}}{\ensuremath{\mathbb{P}(\randVarM = #1)}}
\DeclareDocumentCommand{\ProbVarC}{O{1}}{\ensuremath{\mathbb{P}(\randVarC = #1)}}

\DeclareDocumentCommand{\randVarObjApprox}{}{\ensuremath{Obj_{\textnormal{Alg}}}}
\DeclareDocumentCommand{\optLP}{}{\ensuremath{\textnormal{B}_{\textnormal{LP}}}}
\DeclareDocumentCommand{\optLPstd}{}{\ensuremath{\textnormal{Opt}^{\textnormal{mcf}}_{\textnormal{LP}}}}
\DeclareDocumentCommand{\optLPnew}{}{\ensuremath{\textnormal{Opt}^{\textnormal{new}}_{\textnormal{LP}}}}
\DeclareDocumentCommand{\optIP}{}{\ensuremath{\textnormal{B}_{\textnormal{opt}}}}

\DeclareDocumentCommand{\WAC}{O{\req}}{\ensuremath{\textnormal{WC}_{\req}}}


\DeclareDocumentCommand{\PotEmbeddings}{O{\req}}{\ensuremath{\mathcal{D}_{#1}}}
\DeclareDocumentCommand{\PotEmbeddingsHat}{O{\req}}{\ensuremath{\hat{\mathcal{D}}_{#1}}}



\DeclareDocumentCommand{\maxDemandV}{O{\req} O{\type} O{u}}{\ensuremath{{d}_{\textnormal{max}}(#1,#2,#3)}}
\DeclareDocumentCommand{\maxDemandE}{O{\req} O{u} O{v}}{\ensuremath{d_{\textnormal{max}} ({#1,#2,#3})}}
\DeclareDocumentCommand{\maxDemandX}{O{\req} O{x} O{y}}{\ensuremath{d_{\textnormal{max}} ({#1,#2,#3})}}
\DeclareDocumentCommand{\maxAllocV}{O{\req} O{\type} O{u}}{\ensuremath{{A_{\textnormal{max}}}({#1,#2,#3})}}
\DeclareDocumentCommand{\maxAllocE}{O{\req} O{u} O{v}}{\ensuremath{ {A_{\textnormal{max}}}({#1,#2,#3})}}
\DeclareDocumentCommand{\maxAllocX}{O{\req} O{x} O{y}}{\ensuremath{ {A_{\textnormal{max}}}({#1,#2,#3})}}

\DeclareDocumentCommand{\DeltaV}{}{\ensuremath{\Delta_V}}
\DeclareDocumentCommand{\DeltaV}{}{\ensuremath{\Delta_R}}
\DeclareDocumentCommand{\DeltaE}{}{\ensuremath{\Delta_E}}


\DeclareDocumentCommand{\VVroot}{O{\req}}{\ensuremath{s_{#1}}}
\DeclareDocumentCommand{\VVpred}{O{\req}}{\ensuremath{\pi_{#1}}}

\newcommand{\extractionOrderCharacter}{\ensuremath{\mathcal{X}}}

\newcommand{\VGbfs}[1][\req]{\ensuremath{G^{\extractionOrderCharacter}_{#1}}}
\newcommand{\VVbfs}[1][\req]{\ensuremath{V^{\extractionOrderCharacter}_{#1}}}
\newcommand{\VEbfs}[1][\req]{\ensuremath{E^{\extractionOrderCharacter}_{#1}}}

\DeclareDocumentCommand{\Cycles}{O{\req}}{\ensuremath{\mathcal{C}_{#1}}}
\DeclareDocumentCommand{\Paths}{O{\req}}{\ensuremath{\mathcal{P}_{#1}}}

\DeclareDocumentCommand{\VVcycleSource}{O{\req} O{k}}{\ensuremath{s^{C_{#2}}_{#1}}}
\DeclareDocumentCommand{\VVcycleTarget}{O{\req} O{k}}{\ensuremath{t^{C_{#2}}_{#1}}}
\DeclareDocumentCommand{\VVpathSource}{O{\req} O{k}}{\ensuremath{s^{P_{#2}}_{#1}}}
\DeclareDocumentCommand{\VVpathTarget}{O{\req} O{k}}{\ensuremath{t^{P_{#2}}_{#1}}}

\DeclareDocumentCommand{\VGcycle}{O{\req} O{k}}{\ensuremath{G^{\extractionOrderCharacter,C_{#2}}_{#1}}}
\DeclareDocumentCommand{\VVcycle}{O{\req} O{k}}{\ensuremath{V^{\extractionOrderCharacter,C_{#2}}_{#1}}}
\DeclareDocumentCommand{\VEcycle}{O{\req} O{k}}{\ensuremath{E^{\extractionOrderCharacter,C_{#2}}_{#1}}}
\DeclareDocumentCommand{\VEcycleSame}{O{\req} O{k}}{\ensuremath{\overrightarrow{E}^{C_{#2}}_{#1}}}
\DeclareDocumentCommand{\VEcycleDiff}{O{\req} O{k}}{\ensuremath{\overleftarrow{E}^{C_{#2}}_{#1}}}

\DeclareDocumentCommand{\VGcycleOrig}{O{\req} O{k}}{\ensuremath{G^{C_{#2}}_{#1}}}
\DeclareDocumentCommand{\VVcycleOrig}{O{\req} O{k}}{\ensuremath{V^{C_{#2}}_{#1}}}
\DeclareDocumentCommand{\VEcycleOrig}{O{\req} O{k}}{\ensuremath{E^{C_{#2}}_{#1}}}
\DeclareDocumentCommand{\VGpath}{O{\req} O{k}}{\ensuremath{G^{P_{#2}}_{#1}}}
\DeclareDocumentCommand{\VVpath}{O{\req} O{k}}{\ensuremath{V^{P_{#2}}_{#1}}}
\DeclareDocumentCommand{\VEpath}{O{\req} O{k}}{\ensuremath{E^{P_{#2}}_{#1}}}

\DeclareDocumentCommand{\VEpathSame}{O{\req} O{k}}{\ensuremath{\overrightarrow{E}^{P_{#2}}_{#1}}}
\DeclareDocumentCommand{\VEpathDiff}{O{\req} O{k}}{\ensuremath{\overleftarrow{E}^{P_{#2}}_{#1}}}

\DeclareDocumentCommand{\VEDiff}{O{\req}}{\ensuremath{\overleftarrow{E}^{\extractionOrderCharacter}_{#1}}}

\DeclareDocumentCommand{\VEcycles}{O{\req}}{\ensuremath{E^{\mathcal{C}}_{#1}}}
\DeclareDocumentCommand{\VEpaths}{O{\req}}{\ensuremath{E^{\mathcal{P}}_{#1}}}

\DeclareDocumentCommand{\VVcycleSourcesTargets}{O{\req}}{\ensuremath{V^{\mathcal{C},\pm}_{#1}}}
\DeclareDocumentCommand{\VVpathSourcesTargets}{O{\req}}{\ensuremath{V^{\mathcal{P},\pm}_{#1}}}
\DeclareDocumentCommand{\VVSourcesTargets}{O{\req}}{\ensuremath{V^{\pm}_{#1}}}

\DeclareDocumentCommand{\VVcycleSources}{O{\req}}{\ensuremath{V^{\mathcal{C},+}_{#1}}}
\DeclareDocumentCommand{\VVpathSources}{O{\req}}{\ensuremath{V^{\mathcal{P},+}_{#1}}}

\DeclareDocumentCommand{\VVcycleTargets}{O{\req}}{\ensuremath{V^{\mathcal{C},-}_{#1}}}
\DeclareDocumentCommand{\VVpathTargets}{O{\req}}{\ensuremath{V^{\mathcal{P},-}_{#1}}}

\DeclareDocumentCommand{\VGcycleBranchR}{O{\req} O{k}}{\ensuremath{G^{C_{#2}, B_1}_{#1}}}
\DeclareDocumentCommand{\VVcycleBranchR}{O{\req} O{k}}{\ensuremath{V^{C_{#2}, B_1}_{#1}}}
\DeclareDocumentCommand{\VEcycleBranchR}{O{\req} O{k}}{\ensuremath{E^{C_{#2}, B_1}_{#1}}}

\DeclareDocumentCommand{\VGcycleBranchL}{O{\req} O{k}}{\ensuremath{G^{C_{#2}, B_2}_{#1}}}
\DeclareDocumentCommand{\VVcycleBranchL}{O{\req} O{k}}{\ensuremath{V^{C_{#2}, B_2}_{#1}}}
\DeclareDocumentCommand{\VEcycleBranchL}{O{\req} O{k}}{\ensuremath{E^{C_{#2}, B_2}_{#1}}}

\DeclareDocumentCommand{\VGdecomp}{O{\req} O{k}}{\ensuremath{G^{\mathcal{D}}_{#1}}}
\DeclareDocumentCommand{\VVdecomp}{O{\req} O{k}}{\ensuremath{V^{\mathcal{D}}_{#1}}}
\DeclareDocumentCommand{\VEdecomp}{O{\req} O{k}}{\ensuremath{E^{\mathcal{D}}_{#1}}}

\DeclareDocumentCommand{\VVbranching}{O{\req} }{\ensuremath{\mathcal{B}_{#1}}}
\DeclareDocumentCommand{\VVbranchingcycle}{O{\req} O{k}}{\ensuremath{\mathcal{B}^{C_{#2}}_{#1}}}
\DeclareDocumentCommand{\VVbranchingpath}{O{\req} O{k}}{\ensuremath{\mathcal{B}^{P_{k}}_{#1}}}
\DeclareDocumentCommand{\VVjoin}{O{\req} }{\ensuremath{\mathcal{J}_{#1}}}
\DeclareDocumentCommand{\VVaggregation}{O{\req} }{\ensuremath{\mathcal{A}_{#1}}}

\DeclareDocumentCommand{\VGextcycle}{O{\req} O{k}}{\ensuremath{G^{C_{#2}}_{#1,\textnormal{ext}}}}

\DeclareDocumentCommand{\VVextcycle}{O{\req} O{k}}{\ensuremath{V^{C_{#2}}_{#1,\textnormal{ext}}}}
\DeclareDocumentCommand{\VVextcycleSources}{O{\req} O{k}}{\ensuremath{V^{C_{#2}}_{#1,+}}}
\DeclareDocumentCommand{\VVextcycleTargets}{O{\req} O{k}}{\ensuremath{V^{C_{#2}}_{#1,-}}}
\DeclareDocumentCommand{\VVextcycleSubstrate}{O{\req} O{k}}{\ensuremath{V^{C_{#2}}_{#1,S}}}

\DeclareDocumentCommand{\VEextcycle}{O{\req} O{k}}{\ensuremath{E^{C_{#2}}_{#1,\textnormal{ext}}}}
\DeclareDocumentCommand{\VEextcycleSources}{O{\req} O{k}}{\ensuremath{E^{C_{#2}}_{#1,+}}}
\DeclareDocumentCommand{\VEextcycleTargets}{O{\req} O{k}}{\ensuremath{E^{C_{#2}}_{#1,-}}}
\DeclareDocumentCommand{\VEextcycleSubstrate}{O{\req} O{k}}{\ensuremath{E^{C_{#2}}_{#1,S}}}
\DeclareDocumentCommand{\VEextcycleF}{O{\req} O{k}}{\ensuremath{E^{C_{#2}}_{#1,F}}}

\DeclareDocumentCommand{\VGextpath}{O{\req} O{k}}{\ensuremath{G^{{P_{#2}}}_{#1,\textnormal{ext}}}}

\DeclareDocumentCommand{\VVextpath}{O{\req} O{k}}{\ensuremath{V^{P_{#2}}_{#1,\textnormal{ext}}}}
\DeclareDocumentCommand{\VVextpathSources}{O{\req} O{k}}{\ensuremath{V^{P_{#2}}_{#1,+}}}
\DeclareDocumentCommand{\VVextpathTargets}{O{\req} O{k}}{\ensuremath{V^{P_{#2}}_{#1,-}}}
\DeclareDocumentCommand{\VVextpathSubstrate}{O{\req} O{k}}{\ensuremath{V^{P_{#2}}_{#1,S}}}

\DeclareDocumentCommand{\VEextpath}{O{\req} O{k}}{\ensuremath{E^{P_{#2}}_{#1,\textnormal{ext}}}}
\DeclareDocumentCommand{\VEextpathSources}{O{\req} O{k}}{\ensuremath{E^{P_{#2}}_{#1,+}}}
\DeclareDocumentCommand{\VEextpathTargets}{O{\req} O{k}}{\ensuremath{E^{P_{#2}}_{#1,-}}}
\DeclareDocumentCommand{\VEextpathSubstrate}{O{\req} O{k}}{\ensuremath{E^{P_{#2}}_{#1,S}}}
\DeclareDocumentCommand{\VEextpathF}{O{\req} O{k}}{\ensuremath{E^{P_{#2}}_{#1,F}}}

\DeclareDocumentCommand{\forest}{O{\req}}{\ensuremath{\mathcal{F}_{#1}}}

\DeclareDocumentCommand{\VGforest}{O{\req}}{\ensuremath{G^{\mathcal{A},\mathcal{F}}_{#1}}}
\DeclareDocumentCommand{\VVforest}{O{\req}}{\ensuremath{V^{\mathcal{A},\mathcal{F}}_{#1}}}
\DeclareDocumentCommand{\VEforest}{O{\req}}{\ensuremath{E^{\mathcal{A},\mathcal{F}}_{#1}}}

\DeclareDocumentCommand{\VGforestOrig}{O{\req}}{\ensuremath{G^{\mathcal{F}}_{#1}}}
\DeclareDocumentCommand{\VVforestOrig}{O{\req}}{\ensuremath{V^{\mathcal{F}}_{#1}}}
\DeclareDocumentCommand{\VEforestOrig}{O{\req}}{\ensuremath{E^{\mathcal{F}}_{#1}}}


\DeclareDocumentCommand{\varFlowInput}{O{\req} O{i} O{u}}{\ensuremath{f^+_{#1,#2,#3}}}
\DeclareDocumentCommand{\varFlowOutput}{O{\req} O{i} O{u}}{\ensuremath{f^+_{#1,#2,#3}}}

\DeclareDocumentCommand{\VEextcycleHorizontal}{O{\req} O{k} O{u} O{v}}{\ensuremath{E^{C_{#2}}_{#1,\textnormal{ext},#3,#4}}}
\DeclareDocumentCommand{\VEextpathHorizontal}{O{\req} O{k} O{u} O{v}}{\ensuremath{E^{P_{#2}}_{#1,\textnormal{ext},#3,#4}}}
\DeclareDocumentCommand{\VEextcycleVertical}{O{\req} O{k} O{\type} O{u}}{\ensuremath{E^{C_{#2}}_{#1,\textnormal{ext},#3,#4}}}
\DeclareDocumentCommand{\VEextpathVertical}{O{\req} O{k} O{\type} O{u}}{\ensuremath{E^{P_{#2}}_{#1,\textnormal{ext},#3,#4}}}

\DeclareDocumentCommand{\VEextCGHorizontal}{O{\req} O{u} O{v}}{\ensuremath{E^{\textnormal{ext,SCG}}_{#1,#2,#3}}}
\DeclareDocumentCommand{\VEextCGVertical}{O{\req} O{\type} O{u}}{\ensuremath{E^{\textnormal{ext,SCG}}_{#1,#2,#3}}}

\DeclareDocumentCommand{\VVextCGFlowNodes}{O{\req}}{\ensuremath{V^{\textnormal{ext,SCG}}_{#1,\textnormal{flow}}}}
\DeclareDocumentCommand{\VEextCGFlowEdges}{O{\req}}{\ensuremath{E^{\textnormal{ext,SCG}}_{#1,\textnormal{flow}}}}


\DeclareDocumentCommand{\Queue}{}{\ensuremath{\mathcal{Q}}}
\DeclareDocumentCommand{\QueueC}{}{\ensuremath{\mathcal{Q}_{\mathcal{C}}}}
\DeclareDocumentCommand{\QueueP}{}{\ensuremath{\mathcal{Q}_{\mathcal{P}}}}
\DeclareDocumentCommand{\UsedPaths}{}{\ensuremath{\mathcal{P}}}
\DeclareDocumentCommand{\Variables}{}{\ensuremath{\mathcal{V}}}

\DeclareDocumentCommand{\VGextcycleFlow}{O{\req} O{k}}{\ensuremath{G^{C_{#2}}_{#1,\textnormal{ext},f}}}

\DeclareDocumentCommand{\VGextcycleFlowBranchR}{O{\req} O{k}}{\ensuremath{G^{C_{#2},{B}_1}_{#1,\textnormal{ext},f}}}
\DeclareDocumentCommand{\VVextcycleFlowBranchR}{O{\req} O{k}}{\ensuremath{V^{C_{#2},{B}_1}_{#1,\textnormal{ext},f}}}
\DeclareDocumentCommand{\VEextcycleFlowBranchR}{O{\req} O{k}}{\ensuremath{E^{C_{#2},{B}_1}_{#1,\textnormal{ext},f}}}

\DeclareDocumentCommand{\VGextcycleFlowBranchL}{O{\req} O{k}}{\ensuremath{G^{C_{#2},{B}_2}_{#1,\textnormal{ext},f}}}
\DeclareDocumentCommand{\VVextcycleFlowBranchL}{O{\req} O{k}}{\ensuremath{V^{C_{#2},{B}_2}_{#1,\textnormal{ext},f}}}
\DeclareDocumentCommand{\VEextcycleFlowBranchL}{O{\req} O{k}}{\ensuremath{E^{C_{#2},{B}_2}_{#1,\textnormal{ext},f}}}

\DeclareDocumentCommand{\VGextcycleBranchR}{O{\req} O{k}}{\ensuremath{G^{C_{#2},{B}_1}_{#1,\textnormal{ext}}}}
\DeclareDocumentCommand{\VVextcycleBranchR}{O{\req} O{k}}{\ensuremath{V^{C_{#2},{B}_1}_{#1,\textnormal{ext}}}}
\DeclareDocumentCommand{\VEextcycleBranchR}{O{\req} O{k}}{\ensuremath{E^{C_{#2},{B}_1}_{#1,\textnormal{ext}}}}

\DeclareDocumentCommand{\VGextcycleBranchL}{O{\req} O{k}}{\ensuremath{G^{C_{#2},{B}_2}_{#1,\textnormal{ext}}}}
\DeclareDocumentCommand{\VVextcycleBranchL}{O{\req} O{k}}{\ensuremath{V^{C_{#2},{B}_2}_{#1,\textnormal{ext}}}}
\DeclareDocumentCommand{\VEextcycleBranchL}{O{\req} O{k}}{\ensuremath{E^{C_{#2},{B}_2}_{#1,\textnormal{ext}}}}

\DeclareDocumentCommand{\VGextpathFlow}{O{\req} O{k}}{\ensuremath{G^{P_{#2}}_{#1,\textnormal{ext},f}}}
\DeclareDocumentCommand{\VVextpathFlow}{O{\req} O{k}}{\ensuremath{V^{P_{#2}}_{#1,\textnormal{ext},f}}}
\DeclareDocumentCommand{\VEextpathFlow}{O{\req} O{k}}{\ensuremath{E^{P_{#2}}_{#1,\textnormal{ext},f}}}

\DeclareDocumentCommand{\VVKSource}{O{\req} O{K}}{\ensuremath{s^{K}_{#1}}}
\DeclareDocumentCommand{\VVKTarget}{O{\req} O{K}}{\ensuremath{t^{K}_{#1}}}
\DeclareDocumentCommand{\VVKSourcesTargets}{O{\req}}{\ensuremath{V^{K,\pm}_{#1}}}

\newcommand{\fracSol}{\ensuremath{\mathsf{opt}_{\textnormal{LP}}}}
\newcommand{\intSol}{\ensuremath{\mathsf{opt}_{\textnormal{IP}}}}
\newcommand{\Prob}{\ensuremath{\mathsf{Pr}}}



\DeclareDocumentCommand{\VEbfsPre}{O{j} O{\req}}{\ensuremath{E^{\extractionOrderCharacter,\mathrm{pre}}_{#2,#1}}}
\DeclareDocumentCommand{\VEbfsSuc}{O{i} O{\req}}{\ensuremath{E^{\extractionOrderCharacter,\mathrm{suc}}_{#2,#1}}}

\DeclareDocumentCommand{\VEbfsInter}{O{i} O{j} O{\req}}{\ensuremath{E^{\extractionOrderCharacter}_{#3,#1\leadsto #2}}}

\DeclareDocumentCommand{\VEbfsLabels}{O{e} O{\req} }{\ensuremath{\mathcal{L}^{\extractionOrderCharacter}_{#2,#1}}}

\DeclareDocumentCommand{\VEbfsLabelsOrig}{O{e} O{\req} }{\ensuremath{\mathcal{L}_{#2,#1}}}

\DeclareDocumentCommand{\VEbfsAC}{O{i} O{j} O{\req}}{\ensuremath{C^{\extractionOrderCharacter}_{#1,#2}}}

\DeclareDocumentCommand{\VEbfsACL}{O{i} O{j} O{\req} }{\ensuremath{P^{1}_{#1,#2}}}
\DeclareDocumentCommand{\VEbfsACR}{O{i} O{j} O{\req} }{\ensuremath{P^{2}_{#1,#2}}}

\DeclareDocumentCommand{\VEbfsBags}{O{i} O{\req} }{\ensuremath{\mathcal{B}^{\extractionOrderCharacter}_{#2,#1}}}
\DeclareDocumentCommand{\VEbfsBagIterator}{}{\ensuremath{b}}
\DeclareDocumentCommand{\VEbfsBag}{O{\VEbfsBagIterator} O{i} O{\req}}{\ensuremath{B^{\extractionOrderCharacter,#1}_{#3,#2}}}

\DeclareDocumentCommand{\ewX}{}{\ensuremath{\mathrm{ew}_{\extractionOrderCharacter}}}
\DeclareDocumentCommand{\ew}{}{\ensuremath{\mathrm{ew}}}

\DeclareDocumentCommand{\deltaMinusA}{O{i}}{\ensuremath{\delta^-_{\extractionOrderCharacter}(#1)}}
\DeclareDocumentCommand{\deltaPlusA}{O{i}}{\ensuremath{\delta^+_{\extractionOrderCharacter}(#1)}}

\newcommand{\compP}{\ensuremath{\mathcal{P}}}
\newcommand{\compNP}{\ensuremath{\mathcal{NP}}}
\newcommand{\compPeqNP}{\ensuremath{\compP{\,=\,}\compNP}}
\newcommand{\compPneqNP}{\ensuremath{\compP{\,\neq\,}\compNP}}

\newcommand{\NPhard}{\ensuremath{\compNP\text{-hard}}}
\newcommand{\NPcomplete}{\ensuremath{\compNP\text{-complete}}}
\newcommand{\NPhardness}{\ensuremath{\compNP\text{-hardness}}}
\newcommand{\NPcompleteness}{\ensuremath{\compNP\text{-completeness}}}

\newcommand{\VNEP}{\ensuremath{\textsc{VNEP}}}

\newcommand{\customParagraphStar}[1]{\subsection{{#1}}}


\begin{abstract}
Many resource allocation problems in the cloud can
be described as a basic \emph{Virtual Network Embedding Problem} (VNEP): finding mappings of
\emph{request graphs} (describing the workloads)
onto a \emph{substrate graph} (describing the 
physical infrastructure). In the offline setting, the two natural objectives
are \emph{profit maximization}, i.e., embedding a maximal number of 
request graphs subject to resource constraints,
and \emph{cost minimization}, i.e., embedding all requests 
at minimal overall cost. Hence, the VNEP can be seen as a generalization of classic routing and call admission problems, in which requests are arbitrary graphs whose communication endpoints are not fixed. Due to its applications, the problem has been studied intensively in the networking community. However, the underlying algorithmic problem is hardly understood.

This paper presents the first fixed-parameter tractable approximation
algorithms for the VNEP.
Our algorithms are based on randomized rounding.
Due to the flexible mapping options and the arbitrary request graph topologies,
we show that a novel linear program formulation is required. Only using this novel formulation the computation of convex combinations of valid mappings is enabled, as the formulation needs to account for the structure of the request graphs. Accordingly, to capture the structure of request graphs, we introduce the graph-theoretic notion of extraction orders and extraction width and show that our algorithms have exponential runtime in the request graphs' maximal width. Hence, for request graphs of fixed extraction width, we obtain the first polynomial-time approximations.

Studying the new notion of extraction orders we show that (i) computing extraction orders of minimal width is $\NPhard$ and (ii) that computing decomposable LP solutions is in general $\NPhard$, even when restricting request graphs to planar ones.
\end{abstract}

\maketitle
\allowdisplaybreaks

\section{Introduction}
\label{sec:introduction}

At the heart of the cloud computing paradigm lies the idea of 
efficient resource sharing: due to virtualization, 
multiple workloads can
co-habit and use a given resource infrastructure simultaneously.
Indeed, cloud computing introduces great flexibilities in terms of
\emph{where} workloads can be mapped and accordingly where resources are allocated. At the same time, exploiting this 
mapping flexibility poses a fundamental algorithmic challenge.

The underlying algorithmic problem is essentially a graph theoretical one:
both the workload as well as the infrastructure can be modeled as \emph{graphs}.
The former, the so-called \emph{request graph}, describes the resource requirements
both on the nodes (e.g., the virtual machines) as well as on the interconnecting network.
The latter, the so-called \emph{substrate graph}, describes the physical infrastructure
and its resources (servers and links).

The problem is known in the networking community under the name
\emph{Virtual Network Embedding Problem (VNEP)} and has been studied
intensively in recent years~\cite{vnep,vnep-survey}. The problem arises in many settings, and is also studied in the realm of embedding service chains~\cite{sirocco15,merlin} and virtual clusters~\cite{ccr15emb}.

The online variant in which a minimal cost embedding for a single request is sought after is most prominently studied in the literature. In this work, we study the offline generalization in which multiple requests are given and the objective is to either maximize the profit by selecting a maximal subset of requests to embed or to minimize the cumulative embedding costs. Thus, the offline cost minimization variant reduces to the online problem when considering only a single request.

The design of approximation algorithms for the Virtual Network Embedding Problem has been an open problem
for over a decade~\cite{vnep}.

\subsection{Incarnations of the VNEP in Practice}
\label{app:sec1}
\label{app:sec1:incarnations-of-the-vnep}

To highlight the practical relevance of the VNEP, we
present two examples in  Figure~\ref{fig:vnet-examples-sc-vc}. On the left, a \emph{service chain}  is depicted, which composes existing
network functions (such as a cache, a proxy, or a firewall)
into a more advanced network service. The virtualization of network functions
enables the faster and more flexible allocation in provider networks~\cite{sherry2012making}. 
Concretely, the depicted example is envisioned in the context of 
mobile operators~\cite{ietf-sfc-use-case-mobility-06}: load-balancers (LB$_1$,LB$_2$) are used to
route (parts) of the traffic through 
a cache to optimize the user experience, the firewall (FW) is used to provide security and the network-address translation (NAT) function is used to provide private IP addresses to the customers.
\begin{figure}[tb]
\centering
\includegraphics[width=0.7\columnwidth]{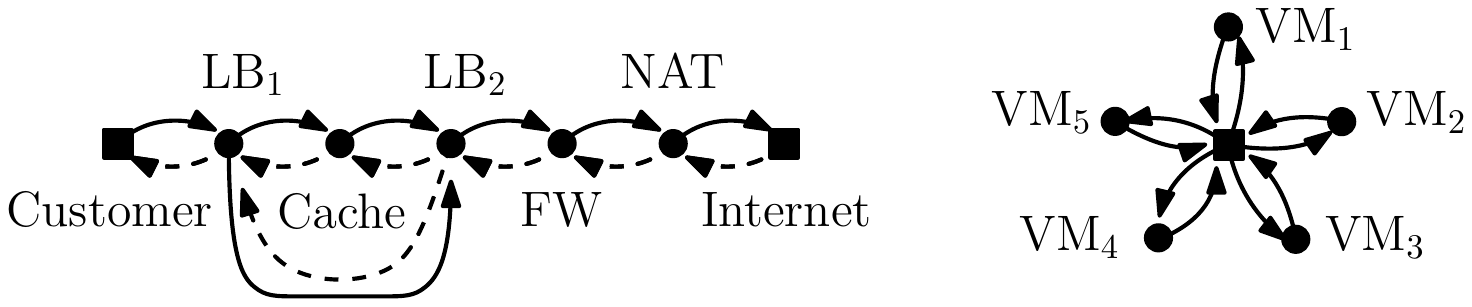}
\caption{Examples for virtual networks (i.e., request graphs). Left: A service chain envisioned in 5G networks~\cite{ietf-sfc-use-case-mobility-06}. Right: a virtual cluster abstraction envisioned in batch processing applications~\cite{oktopus}.}
\label{fig:vnet-examples-sc-vc}
\end{figure}

Depicted on the right of Figure~\ref{fig:vnet-examples-sc-vc} is a \emph{virtual cluster}, which was proposed as an abstraction for batch processing applications
in the cloud~\cite{oktopus}. Concretely, a virtual cluster consists of a set of virtual machines (VMs) and a  single logical switch which connects all virtual machines. As originally proposed, all virtual machines and all links have the same  computational and bandwidth demands, respectively.  The abstraction is attractive due to its simplicity: users only need to specify three numbers, namely the number of VMs together with their uniform demands and the bandwidth to the logical switch. 

\subsection{Problem Statement}
Given is a physical  network $\SG=(\SV,\SE)$ offering a set~$\types$ of computational types. We refer to the physical network as the \emph{substrate} network. For a type $\type \in \types$, the set $\SVTypes \subseteq \SV$ denotes the substrate nodes that can host functionality of the type $\type$. Denoting the node resources by $\SRV = \{(\type, u)~| \type \in \types, u \in \SVTypes\}$ and all substrate resources by $\SR = \SRV \cup \SE$, the capacity of nodes and edges is denoted by $\Scap(x,y) > 0$ for  $(x,y) \in \SR$.

The set of request is denoted by $\requests$. 
For each request $\req \in \requests$, a directed graph $\VG=(\VV,\VE)$ is given. We refer to the nodes of these graphs as virtual or request nodes and to the edges as virtual or request edges.
The type of a virtual node is given via the function $\Vtype : \VV \to \types$. We allow for node and edge mapping restrictions. Concretely, we assume that the mapping of a virtual node $i \in \VV$ is restricted to a set $\VVloc \subseteq \SVTypes[\Vtype(i)]$, while the mapping of a virtual edge $(i,j)$ is restricted to a set $\VEloc \subseteq \SE$. Each virtual node $i \in \VV$ and each edge $(i,j) \in \VE$ is attributed with a resource demand $\Vcap(i) \geq 0$ and $\Vcap(i,j) \geq 0$, respectively.
Virtual nodes and edges can only be mapped on substrate nodes and edges of sufficient capacity and we have $\VVloc \subseteq \{u \in \SVTypes[\Vtype(i)] | \Scap(u) \geq \Vcap(i)\}$ and $\VEloc \subseteq \{(u,v) \in \SE | \Scap(u,v) \geq \Vcap(i,j)\}$. We denote by $d_{\max}(r,x,y)$ the maximal demand that a request $\req$ may impose on a resource $(x,y) \in \SR$, i.e. $\maxDemandV  =  \max_{i \in \VV: u \in \VVloc} \Vcap(i)$ for $(\tau,u) \in \SRV$ and $\maxDemandE = \max_{(i,j) \in  \VE: (u,v) \in \VEloc} \Vcap(i,j)$ for $(u,v) \in \SE$.

The mapping of a request onto the substrate graph is captured by the following definition.

\begin{definition}[Valid Mapping]
\label{def:valid-mapping}
A valid mapping~$\map$ of request~$\req \in \requests$ is a tuple~$(\mapV, \mapE)$ of functions $\mapV : \VV \to \SV$ and $\mapE : \VE \to \mathcal{P}(\SE)$, such that the following conditions hold:
\begin{itemize}
\item Virtual nodes are mapped to allowed substrate nodes: $\mapV(i) \in  \VVloc$  holds for all $i \in \VV$.
\item The mapping $\mapE(i,j)$ of virtual edge $(i,j) \in \VE$ is a path connecting~$\mapV(i)$ to $\mapV(j)$, which only uses allowed edges, i.e., $\mapE(i,j) \subseteq \mathcal{P}(\VEloc)$ holds. 
\end{itemize} 
We denote by $\spaceSolReq$ the set of \emph{all} valid mappings of request $\req \in \requests$. 
\end{definition}

Note that the edge mapping $\mapE(i,j)$ may be empty, iff. $\mapV(i) = \mapV(j)$ holds for edges $(i,j) \in \VE$. Next, we introduce the notion of allocations induced by a valid mapping.

\begin{definition}[Allocations of a Valid Mapping]
The allocation $A(\map,x,y)$ induced by mapping $\map$ 
on resource $(x,y) \in \SR$ is defined as follows: $A(\map, \tau,u)= \sum_{i \in \VV, \type(i)=\tau, \mapV(i)=u } \Vcap(i)$ holds for $(\tau,u) \in \SRV$ and $A(\map,u,v) = \sum_{ (i,j)\in \VE, (u,v) \in \mapE(i,j)}   \Vcap(i,j)$ holds for $(u,v) \in \SE$, respectively. 
The maximal allocation that a valid mapping of request $\req$ may impose on a substrate resource $(x,y) \in \SR$ is denoted by $\maxAllocX  =  \max_{\mapping \in \spaceSolReq} A(\map,x,y)$.
\end{definition}

Given a collection of valid mappings $\{\map\}_{\req \in \requestsP}$ for a subset of requests $\requestsP \subseteq \requests$, we refer to this collection as \emph{feasible} if the cumulative allocations obey substrate capacities, i.e., $\sum_{\req \in \requestsP} A(\map,x,y) \leq \Scap(x,y)$ holds for all resources $(x,y) \in \SR$. 

\begin{definition}[Virtual Network Embedding Problem]
The profit variant of the Virtual Network Embedding Problem (VNEP) asks for finding a feasible collection $\{\map\}_{\req \in \requestsP}$ of mappings while maximizing the overall profit $\sum_{\req \in \requestsP} \Vprofit$, where $\Vprofit > 0$ denotes the benefit obtained for embedding request $\req$. 
For the cost variant all of the given requests $\requests$ must be feasibly embedded while minimizing the resource costs $\sum_{(x,y) \in \SR} \Scost(x,y) \cdot \sum_{\req \in \requests} A(\map, x,y)$, where $\Scost(x,y) \geq 0$ denotes the resource cost of $(x,y) \in \SR$.
\end{definition}

\begin{figure}[h]

 {
  \LinesNotNumbered
  \renewcommand{\arraystretch}{0}
 
 \removelatexerror

  \begin{IPFormulation}{H}
  
  \popline
 
  \SetAlgorithmName{Formulation}{}{{}}

  \newcommand{\spaceIt}{\qquad\quad\quad}
  \newcommand{\miniSpace}{\hspace{1.5pt}}

\noindent

  \begin{tabular}{LFRLQBFRLQB}
  \multicolumn{5}{r}{\parbox{0.42\textwidth}{~}} & \multicolumn{6}{r}{\parbox{0.52\textwidth}{~}} \\[-2pt]

\hspace{-4pt}\tagIt{alg:exp:profit:obj} 
~\hspace{0.2cm}~ &   \multicolumn{4}{C}{\textnormal{max~}  \sum \limits_{\req \in \requests, \map \in \spaceSolReq} \hspace{-12pt} \prob \cdot \Vprofit   } &    &  \multicolumn{4}{C}{\textnormal{min~}  \sum_{(x,y) \in \SR} \hspace{-4pt}  \Scost(x,y) \sum \limits_{\req \in \requests, \map \in \spaceSolReq} \hspace{-12pt}  \prob \cdot A(\map,x,y)  } &   \tagIt{alg:exp:cost:obj}\\
   
    	
\hspace{-4pt}\tagIt{alg:exp:profit:embedding-choice} & ~\hspace{1.0cm}~ \sum \limits_{\mapping \in \spaceSolReq} \prob  & \leq  & 1 ~~~~& \forall \req \in \requests &  & ~\hspace{1.1cm}~\sum \limits_{\mapping \in \spaceSolReq} \prob  & =  & 1 ~~~~& \forall \req \in \requests &  ~\hspace{0.2cm}~\tagIt{alg:exp:cost:embedding-choice} \\[20pt]

\cline{2-10} \\[4pt]
    	
 	\multicolumn{9}{C}{\sum \limits_{\req \in \requests, \mapping \in \spaceSolReq} \prob \cdot A(\mapping, x,y,)   \leq   \Scap(x,y) } & \forall (x,y) \in  \SR & \tagIt{alg:exp:capacities} 
 	
%
 	 	
\end{tabular}
  \caption{Enumerative Formulation for the VNEP (left: profit, right: cost)}
  \label{alg:mapping-formulation}
  \end{IPFormulation}
  }
\end{figure}

The VNEP can be expressed as the non-polynomial sized Formulation~\ref{alg:mapping-formulation}, which explicitly enumerates \emph{all} valid mappings: for each request $\req \in \requests$ and each mapping $\mapping \in \spaceSolReq$ a variable $\prob$ is introduced. Setting $\prob \in \{0,1\}$ yields integer programs and setting $\prob \in [0,1]$ yields the respective linear programs. We refer to the problem over the linear variables, in which convex combinations of valid mappings are allowed, as the \emph{fractional} VNEP.

The Virtual Network Embedding Problem is known to be strongly $\NPhard$~\cite{rostSchmidVNEPComplexity}.

\subsection{Putting the VNEP Into Perspective}
The VNEP can be seen as a generalization of many well-studied problems. The profit variant is e.g. related to \emph{routing requests}~\cite{awerbuch1993throughput,bartal1997line} and  \emph{virtual circuits}~\cite{aspnes1997line,plotkin1995competitive}, and \emph{the unsplittable flow problem}~\cite{bansal2009logarithmic}, while the cost variant is related to the shortest $k$-disjoint paths problem~\cite{bjorklund2014shortest,robertson1995graph}, and the subgraph isomorphism problem~\cite{eppstein1995subgraph}. 

The most notable differences to the aforementioned problems are (i) that request node locations are not fixed a priori and that (ii) a single request represents a graph instead of e.g. a single link as in the unsplittable flow problem. Accordingly, the key challenge we face when designing approximation algorithms, is that virtual nodes can in principle be mapped on \emph{any}
substrate node and each virtual edge may traverse
any substrate edge.

\subsection{Our Results and Techniques}
In this paper we set out to initiate the study of approximation 
algorithms for the VNEP for \emph{arbitrary} request graphs. 
Leverging the VNEP's connection to multi-commodity flow problems, 
we employ randomized rounding to obtain our results.
This technique has proven both simple and effective: 
given an Integer Program (IP) for a problem, 
solutions of its Linear Program (LP) are decomposed into convex combinations (cf. Formulation~\ref{alg:mapping-formulation}) and then \emph{rounded} according to their weight.

While in many contexts the natural LP, obtained by relaxing the corresponding integer program, is sufficiently 
strong to extract convex combinations of solutions, this is not the case for the VNEP. 
As the mapping of flow endpoints is flexible in the VNEP, 
we prove that the natural Multi-Commodity Flow (MCF) formulation for the 
VNEP fails to ensure the decomposability into convex combinations. 
In fact, it fails to capture the structure of valid mappings and we prove that the MCF formulation's integrality gap is unbounded. 

Analyzing the shortcomings of the MCF formulation, 
we obtain sufficient conditions to ensure decomposability.
Accordingly, we develop a novel LP formulation for the VNEP 
which incorporates the requests' individual structure. 
The dependency of our formulation on the underlying request graphs 
comes at the price that the size of the formulation grows exponentially
in the `complexity' of the request graphs. Our formulation relies on acyclic \mbox{(re-)orientations} of request graphs called \emph{extraction orders} to guide the process of extracting valid mappings. Based on confluences in these extraction orders, i.e. disjoint paths, we introduce the notion of \emph{extraction width}. In turn, we show that the size of our LP formulation, and hence also the runtime of our approximations, 
are fixed-parameter tractable (FPT) in the \emph{extraction width} of the
given \emph{extraction orders}.

Hence, finding efficient approximations boils down to finding extraction orders 
of small width. Our initial results are quite intriguing: we show that depending on the chosen extraction order the width can differ by a factor of $\Omega(|\VG|)$ (which is maximal) and that finding the minimal extraction width is itself $\NPhard$. While this may raise questions about the sensibility of our graph-theoretic notions, we also show that there cannot exist any polynomial-time algorithm (neither linear nor combinatorial) that can solve the fractional VNEP (even when restricting the requests to planar graphs).

Having set out to obtain approximations for the VNEP, we eventually derive the first (FPT-)approximations for the profit and cost variants of the VNEP for arbitrary request graphs  by using our novel LP formulation. The presented approximations provide constant approximation guarantees for the cost and the profit while exceeding resource capacities by a factor of $\mathcal{O}(1 + \varepsilon \cdot \sqrt{2\cdot  \Delta(\SR) \cdot \log |\SV|})$, where $\varepsilon \leq 1$ is the ratio of maximal demand to minimal capacity and $\Delta(\SR) = \max_{(x,y) \in \SR} \sum_{\req \in \requests} \left(\frac{\maxAllocX}{ \maxDemandX} \right)^2$ captures the (sum of squared) ratios of the maximal cumulative allocation divided by the maximal allocation.

\subsection{Related Work}
In the last decade, the VNEP has attracted much attention due to its numerous applications.
A survey from 2013 lists more than 80 different algorithms for its many 
variations~\cite{vnep-survey}. 
A large fraction of the existing literature 
considers heuristics without giving approximation 
guarantees~\cite{vnep,vnep-rethink}. 
Other works proposed exact methods as integer or constraint programming, 
coming at the cost of an exponential runtime~\cite{HVSBFTF15,rostSchmidFeldmann2014}.

In contrast, we initiate the study of (FPT-)approximation algorithms for the VNEP with provable approximation guarantees for \emph{arbitrary substrate and request graphs}.
The works closest to ours are by Even et al.~\cite{sss-guy,sirocco16path} and Bansal et al.~\cite{bansal2011minimum}. Even et al. studied approximation algorithms and competitive online algorithms for the embedding of request chains. Bansal et al. consider the setting of embedding tree request graphs under the objective to minimize the maximum load and also provide approximations and competitive online algorithms.
Their main result is a $n^{O(d)}$ time~$O(d^2 \log{(nd)})$-approximation algorithm for the embedding of a single tree of depth $d$ on a substrate with $n$ nodes, which is based on a strong LP relaxation
inspired by the Sherali-Adams hierarchy. By considering only tree requests, Bansal et al. do not address the problem of computing valid mappings for request graphs containing cycles. However, and importantly, the approach of Bansal et al. is complementary to ours and may hence potentially be combined with our results in the future to obtain stronger approximations and also derive competitive online algorithms.

\noindent\textbf{Bibliographic Note.} This work significantly extends the authors' previous technical report~\cite{rostSchmidRandomizedRoundingFirstTR2016} as well as the publication~\cite{rostSchmidIFIPLeveragingRandomizedRoundingWithPreprint}, which only consider approximation algorithms for cactus request graphs and are hence not applicable for arbitrary request graphs.

\subsection{Organization}
The remainder of this paper is organized as follows.
Section~\ref{sec:classic-mcf-and-its-limits} studies
the classic multi-commodity flow formulation 
and shows its limitations.
In Section~\ref{sec:decomp} we present our decomposable
LP formulation and introduce graph-theoretic notions as extraction confluences and extraction width. 
In Section~\ref{sec:approximation-via-randround}
we present our FPT-approximations for the VNEP.
In Section~\ref{sec:novel-formulation:size-of-formulation-and-graph-classifications} we shortly study properties of the novel extraction width concept and show that cactus request graphs have a constant extraction width.
We conclude our paper in Section~\ref{sec:conclusion}.

\section{Limitations of Classic Multi-Commodity Formulations for VNEP}
\label{sec:classic-mcf-and-its-limits}

In this section, we study the Multi-Commodity Flow (MCF) formulation for solving the VNEP (see Formulation~\ref{alg:VNEP-IP-old}), which is widely used~\cite{vnep,mehraghdam2014specifying,rostSchmidVNEPComplexity,rostSchmidFeldmann2014}. We first show  the positive result that the formulation is sufficiently strong to compute solutions to the fractional VNEP when requests are \emph{trees}. Subsequently, we show that the formulation fails to allow for the decomposition of \emph{cyclic} request graphs into convex combinations of valid mappings.

\begin{figure}[t!]

 {
  \LinesNotNumbered
  \renewcommand{\arraystretch}{0.0}
 
 \removelatexerror

  \begin{IPFormulation}{H}
  
  \popline
 
  \SetAlgorithmName{Formulation}{}{{}}

  \newcommand{\spaceIt}{\qquad\quad\quad}
  \newcommand{\miniSpace}{\hspace{1.5pt}}
 
  \begin{tabular}{FRLQB}
  \multicolumn{5}{r}{\parbox{0.975\textwidth}{~}} \\[-2pt]

   \sum \limits_{u \in \VVloc} y^u_{\req, i} & = & x_{\req} & \forall \req \in \requests, i \in \VV   &  \tagIt{alg:VNEP-old:node-embedding} \\
 y^u_{\req, i} & = & 0 & \forall \req \in \requests, i \in \VV, u \in \SV \setminus \VVloc &  \tagIt{alg:VNEP-old:node-embedding-non-mappable} \\[8pt]
   \sum \limits_{(u,v) \in  \delta^+(u)} z^{u,v}_{\req,i,j} - \sum \limits_{(v,u) \in  \delta^-(u)}  z^{v,u}_{\req,i,j} & = & y^u_{\req, i} - y^u_{\req,j} ~~ \quad \quad\quad &  \forall \req \in \requests, (i,j) \in  \VE, u \in \SV &  \tagIt{alg:VNEP-old:edge-embedding}\\
   
      z^{u,v}_{\req,i,j} & = & 0 &  \forall \req \in \requests, (i,j) \in  \VE, (u,v) \in \SE \setminus \VEloc &  \tagIt{alg:VNEP-old:edge-embedding-non-mappable}\\[8pt]
      
      	\sum \limits_{(i,j) \in  \VE } \Vcap(i,j) \cdot z^{u,v}_{\req,i,j} & = &  a^{u,v}_{\req} & \forall \req \in \requests, (u,v) \in  \SE & \tagIt{alg:VNEP-old:load-edge} \\

 	\sum \limits_{i \in \VV, \Vtype(i) = \type}  \Vcap(i) \cdot y^u_{\req,i}  & =  & a^{\type,u}_{\req} & \forall \req \in \requests, (\type,u) \in  \SRV &  \tagIt{alg:VNEP-old:load-node}\\
 	
 	\sum \limits_{\req \in \requests} a^{x,y}_{\req}  & \leq  & \Scap(x,y) & \forall (x,y) \in  \SR &  \tagIt{alg:VNEP-old:capacities} 
%
%
  \end{tabular}
  \caption{Multi-Commodity Flow Base Formulation for the VNEP}
  \label{alg:VNEP-IP-old}
  \end{IPFormulation}
  }
\end{figure}

\subsection{The Multi-Commodity Formulation}
We explain the formulation by considering its integer variant. The variable $x_{\req} \in \{0,1\}$ indicates whether request $\req \in \requests$ is embedded or not. The variable $y^u_{\req,i} \in \{0,1\}$ indicates whether virtual node $i \in \VV$ is mapped on substrate node $u$. Similarly, the flow variable $z^{u,v}_{\req,i,j} \in \{0,1\}$ indicates whether the substrate edge $(u,v) \in \SE$ is part of the path of the virtual edge $(i,j) \in \VE$. The variable $a^{x,y}_{\req} \geq 0$ denotes the cumulative allocations that the embedding of request $\req$ induces on resource $(x,y) \in \SR$.

By Constraints~\ref{alg:VNEP-old:node-embedding} and Constraint~\ref{alg:VNEP-old:node-embedding-non-mappable}, virtual nodes are only mapped on suitable substrate nodes when $x_\req = 1$ holds. Constraint~\ref{alg:VNEP-old:edge-embedding} induces an unsplittable unit flow for each virtual edge $(i,j) \in  \VE$ from the substrate location onto which $i$ was mapped to the substrate location onto which $j$ was mapped. By Constraint~\ref{alg:VNEP-old:edge-embedding-non-mappable} the mapping of virtual edges may only consist of \emph{allowed} substrate edges.
Constraints~\ref{alg:VNEP-old:load-edge} and~\ref{alg:VNEP-old:load-node} compute the cumulative allocations while Constraint~\ref{alg:VNEP-old:capacities} enforces that resource capacities are respected. 
Applying the objective $\max \sum_{\req \in \requests} \Vprofit \cdot x_{\req}$ the profit variant is obtained. Setting $\min \sum_{(x,y) \in \SR, \req \in \requests} \Scost(x,y) \cdot a^{x,y}_{\req}$ and enforcing $x_{\req} = 1$ for all requests $\req \in \requests$ the cost variant is obtained.

The LP formulation is obtained by relaxing the domain of the above introduced binary variables to $[0,1]$. The following lemma states that whenever a virtual node $i \in \VV$ is (fractionally) mapped on a certain substrate node, suitable mappings for all incident edges and their endpoints can be found. 

\begin{restatable}[Local Connectivity Property of the MCF Formulation]{lemm}{localConnectivity}~\\
\label{lem:local-connectivity-property}
Consider a fractional solution $(x_{\req},\vec{y}_{\req},\vec{z}_{\req},\vec{a}_{\req})$ to the LP Formulation~\ref{alg:VNEP-IP-old} for request $\req \in \requests$.
If $y^u_{\req,i} > 0$ holds for  $i \in \VV$ and $u \in  \VVloc$, then for incoming edges $(k,i) \in \VE$ and outgoing edges $(i,j) \in \VE$ there exist substrate paths $P^{v,u}_{\req,k,i}$ and $P^{u,w}_{\req,i,j}$, such that:
\begin{enumerate}
\item $P^{v,u}_{\req,k,i}$ is a path from $v$ to $u$, such that $y^v_{\req,k} > 0$ and $z^{e}_{\req,k,i} > 0$ holds for $e \in P^{v,u}_{\req,k,i}$.
\item $P^{u,w}_{\req,i,j}$ is a path from $u$ to $w$, such that $y^w_{\req,j} > 0$ and $z^{e}_{\req,i,j} > 0$ holds for $e \in P^{u,w}_{\req,i,j}$.
\end{enumerate}
The respective paths $P^{v,u}_{\req,k,i}$ and $P^{u,w}_{\req,i,j}$ can be found in time $\mathcal{O}(|\SE|)$ by a simple graph search.
\end{restatable}
\begin{proof}
Fix any substrate node $u \in \SV$ for which $y^u_{\req,i} > 0$ holds. We first consider the outgoing edges $(i,j) \in \VE$.  By Constraint~\ref{alg:VNEP-old:node-embedding}, $\sum_{u \in \VVloc} y^u_{\req,i} = \sum_{v \in \VVloc[j]} y^v_{\req,j}$ holds. Hence, the virtual node $j \in \VV$ must be mapped also at least with value $y^u_{\req,i}$. If $j$ is also partially mapped on $u$, i.e., if $y^u_{\req,j} > 0$ holds, then the result follows directly, as $u$ connects to $u$ using (and allowing) the empty path $P^{u,u}_{\req,i,j} = \langle \rangle$.
If, on the other hand, $y^u_{\req,j} = 0$ holds, then Constraint~\ref{alg:VNEP-old:edge-embedding} induces an flow of value $y^u_{\req,i}$ at substrate node $u$ with respect to the commodity $z_{\req,i,j}$. As the right hand side of Constraint~\ref{alg:VNEP-old:edge-embedding} may only attain negative  values at nodes $w \in \VVloc[j]$ for which $y^w_{\req,j} > 0$ holds, the flow (of commodity $z_{\req,i,j}$) emitted at node $u$ must eventually reach a node $w \in \VVloc[j]$ with $y^w_{\req,j} > 0$ and hence the result follows for any outgoing edge $(i,j) \in \VE$.
Note that the corresponding path $P^{u,w}_{\req,i,j}$ can be constructed in time $\mathcal{O}(|\SE|)$ by a simple breadth-first search, which only considers edges $(u',v') \in \SE$ for which $z^{u',v'}_{\req,i,j} > 0$ holds.

The argument for incoming edges $(k,i) \in \VE$ is the same and the respective paths $P^{v,u}_{\req,k,i}$ can be recovered by breadth-first searches traversing substrate edges $(u,v) \in \SE$ in their opposite direction when $z^{u,v}_{\req,k,i} > 0$ holds.
 \end{proof}

\subsection{Decomposing Solutions to the MCF Formulation}
\label{sec:decomposing-mcf-solution-trees}
Given the connectivity property of Lemma~\ref{lem:local-connectivity-property}, we argue how solutions to the LP relaxation of the MCF formulation can be decomposed into convex combinations $\PotEmbeddings = \{(\prob,\mapping)\}_k$ (cf.~LP Formulation~\ref{alg:mapping-formulation}) \emph{as long as the request graphs are trees}. The ideas presented henceforth will also apply for the decomposition of our novel formulation presented in Section~\ref{sec:decomp}.

\DeclareDocumentCommand{\origEdge}{O{e}}{\ensuremath{\overrightarrow{E}_{\hspace{-1pt}\req}\hspace{-1pt} (#1)}}
\DeclareDocumentCommand{\origEdgeFun}{}{\ensuremath{ \overrightarrow{E}_{\req} }}

\DeclareDocumentCommand{\extractionEdge}{O{e}}{\ensuremath{\overrightarrow{E}^{\hspace{-1pt}\extractionOrderCharacter}_{\hspace{-1pt}\req}\hspace{-2pt}(#1)}}
\DeclareDocumentCommand{\extractionEdgeFun}{}{\ensuremath{ \overrightarrow{E}^{\hspace{-1pt}\extractionOrderCharacter}_{\hspace{-1pt}\req} }}

We naturally apply the idea of Ford and Fulkerson~\cite{ford1962flows} for decomposing $s-t$ flows into paths to our setting. Given a LP solution $(x_{\req},\vec{y}_{\req}, \vec{z}_{\req},\vec{a}_{\req})$ for request $\req \in \requests$, we need to find a valid mapping $\map=(\mapV,\mapE) \in \spaceSolReq$ which is \emph{covered} by the embedding variables. Concretely, letting $\mathcal{V}(\map) = \{y^{\mapV(i)}_{\req,i}| i\in \VV \} \cup \{z^{u,v}_{\req,i,j} | (i,j) \in \VE, (u,v) \in \mapE(i,j) \}$ denote all the LP variables involved under mapping $\map$, we say that the mapping $\map$ is covered by the LP solution iff. $f_r = \min \mathcal{V} = > 0$ holds. Accordingly, the mapping $\map$ of weight $f_{\req}$ can be \emph{extracted} by reducing the variables in $\mathcal{V}$ by $f_{\req}$ while adding $(f_{\req}, \map)$ to the set of convex combinations $\PotEmbeddings$. Importantly, after the extraction, the now adapted LP solution is still feasible and hence the extraction process can be repeated. To find a mapping in the first place, the mapping of nodes and edges has to be done in some order. We refer to this order as the extraction order:

\begin{definition}[Extraction Order $\VGbfs$]
\label{def:orientation-graph}
Given a virtual network $\VG=(\VV,\VE)$, we refer to any rooted graph $\VGbfs = (\VV,\VEbfs,\VVroot)$ as an \emph{extraction order}, if the following holds:
\begin{enumerate}
\item $\VGbfs$ is a directed acyclic graph, s.t. each node is reachable from the root $\VVroot \in \VV$, and
\item $\VEbfs$ is obtained from $\VE$ by (potentially) reversing the orientation of some edges.
\end{enumerate}
We denote by $\origEdgeFun: \VEbfs \to \VE$ the function yielding the edge's original orientation and by $\extractionEdgeFun: \VE \to \VEbfs$ its inverse. We write  $\deltaPlusA = \{(i,j) \in \VEbfs \}$ and $\deltaMinusA = \{(j,i) \in \VEbfs \}$ to denote the outgoing and incoming edges with respect to the edge set $\VEbfs$.
\end{definition}

Given the extraction order $\VGbfs$, the extraction process works by first choosing a suitable mapping location for the root $\VVroot$. Given this location, Lemma~\ref{lem:local-connectivity-property} is applied to obtain mappings for all outgoing edges of $\VVroot$ (according to $\VEbfs$) \emph{together} with their heads. Continuing to apply Lemma~\ref{lem:local-connectivity-property} for each of the newly mapped nodes, a complete mapping in which all virtual nodes and edges are mapped on suitable substrate nodes and edges is constructed.

\begin{figure}[t!]
 
 \scalebox{0.92}{
 \begin{minipage}{1.07\columnwidth}
 
 \removelatexerror
 
 \begin{algorithm*}[H]

 \SetKwInOut{Input}{Input}\SetKwInOut{Output}{Output}
 \SetKwFunction{ProcessPath}{ProcessPath}{}{}
 \SetKwFunction{reverse}{reverse}{}{}
 \SetKwFunction{LP}{LP}
 \SetKwFunction{LP}{LP}
 
 \newcommand{\SET}{\textbf{set~}}
 \newcommand{\ADD}{\textbf{add~}}
 \newcommand{\DEFINE}{\textbf{define~}}
 \newcommand{\AND}{\textbf{and~}}
 \newcommand{\LET}{\textbf{let~}}
 \newcommand{\WITH}{\textbf{with~}}
 \newcommand{\COMPUTE}{\textbf{compute~}}
 \newcommand{\FIND}{\textbf{find~}}
 \newcommand{\CHOOSE}{\textbf{choose~}}
 \newcommand{\DECOMPOSE}{\textbf{decompose~}}
 \newcommand{\FORALL}{\textbf{for all~}}
 \newcommand{\OBTAIN}{\textbf{obtain~}}
 \newcommand{\WITHPROBABILITY}{\textbf{with probability~}}

 \Input{Tree request $\req \in \requests$ together with a solution~$(x_{\req},\vec{y}_{\req},\vec{z}_{\req},\vec{a}_{\req})$ for Formulation~\ref{alg:VNEP-IP-old}\\
 Extraction order $\VGbfs = (\VV,\VEbfs, \VVroot)$ }
 \Output{Convex combination~$\PotEmbeddings = \{\decomp = (\prob,\mapping)\}_k$ of valid mappings}
 \BlankLine
 
 
 	\SET $\PotEmbeddings  \gets \emptyset$ \AND $k \gets 1$\\
 	\While{$x_{\req} > 0$ }
 	{
 		
 		\SET $\mapping \gets (\mapV,\mapE)~\gets (\emptyset,\emptyset)$ \label{alg:sc-decomposition:init-maps}\\
 		
 		\SET $\Queue \gets \{\VVroot \}$\\
 		
 		\CHOOSE $u \in \VVloc[\VVroot]$ \WITH $y^u_{\req,\VVroot} > 0$ \AND \SET $\mapV(\VVroot)~\gets u$\\
 		
 		\While{$|\Queue| > 0$}{	\label{alg:decomposition:begin-while-q}
 			\CHOOSE $i \in \Queue$ \AND \SET$\Queue \gets \Queue \setminus \{i\}$\\
 			\ForEach{$(i,j) \in \delta^+_{\VEbfs}(i)$}{
 				\eIf{$(i,j) = \origEdge[i,j]$}{
 					\COMPUTE path ${P}^{u,v}_{\req,i,j}$ from $\mapV(i)=u$ to  $v \in \VVloc[j]$ according to Lemma~\ref{lem:local-connectivity-property}\\
 				\pushline\pushline\pushline \nonl such that $y^v_{\req,j} > 0$ and 
 				  $z^{u',v'}_{\req,i,j} > 0$ hold for $(u',v') \in {P}^{u,v}_{\req,i,j}$\\
 				  					\vspace{2pt}
 				  \popline\popline\popline
 				  
 					\SET $\mapV(j) \gets v$ \AND $\mapE(i,j) \gets {P}^{u,v}_{\req,i,j}$\\
 				}{
 					\COMPUTE path ${P}^{v,u}_{\req,j,i}$ from $v \in \VVloc[j]$ to $\mapV(i)=u$ according to Lemma~\ref{lem:local-connectivity-property}\\	\pushline\pushline\pushline \nonl such that $y^v_{\req,j} > 0$ and 
 									  $z^{u',v'}_{\req,j,i} > 0$ hold for $(u',v') \in {P}^{v,u}_{\req,j,i}$\\
 					\vspace{2pt}
 					\popline\popline\popline
 					\SET $\mapV(j) \gets v$ \AND  $\mapE(\origEdge[i,j]) \gets {P}^{u,v}_{\req,j,i} $\\

 				}
 				\SET $\mathcal{Q} \gets \mathcal{Q} \cup \{j\}$\\
 			}
 		}
 		\SET $\mathcal{V}_k \gets \{x_\req\} \cup \{y^{\mapV(i)}_{\req,i} | i \in \VV\} \cup \{z^{u,v}_{\req,i,j} | (i,j) \in \VE, (u,v) \in \mapE(i,j)\}$ \label{alg:decomposition:compute-Vk}\\
 		\SET $\prob \gets \min \mathcal{V}_k$ \label{alg:decomposition:computing-prob} \\
 		\SET $v \gets v - \prob$ \FORALL $v \in \mathcal{V}_k$ \AND \SET $a^{x,y}_{\req} \gets a^{x,y}_{\req} - \prob \cdot A(\mapping,x,y)$ \FORALL $(x,y) \in \SR$ \label{alg:decomposition:adapt-variables}\\	
 		\ADD $\decomp = (f^k_{\req},m^k_{\req})$ to $\PotEmbeddings$ \AND \SET $k \gets k + 1$\\
 	}

 \KwRet{$\PotEmbeddings$}
 \caption{Decomposition algorithm of MCF solutions for Tree Requests}
 \label{alg:decompositionAlgorithm-MCF-Tree}
 \end{algorithm*}
 \end{minipage}}
 
 \end{figure}

Algorithm~\ref{alg:decompositionAlgorithm-MCF-Tree} formalizes the decomposition scheme to extract convex combinations of valid mappings from solutions to Formulation~\ref{alg:VNEP-IP-old}.
 The algorithm extracts mappings $\mapping$ of value $\prob$ iteratively, as long as $x_{\req} > 0$ holds. Initially, in the $k$-th iteration, none of the virtual nodes and edges are mapped. As $x_\req > 0$ holds, the root node $\VVroot$ must be mapped accordingly by Constraint~\ref{alg:VNEP-old:node-embedding}, i.e. there must exist a node $u \in \VVloc[\VVroot]$ with $y^{u}_{\req,\VVroot} > 0$ and the algorithm sets $\mapV(\VVroot) = u$.  Given this initial fixing, the algorithm iteratively extracts nodes from the queue $\mathcal{Q}$ which have already been mapped and considers all outgoing  virtual edges $(i,j) \in \VEbfs$. If the orientation of edge $(i,j)$ was not changed, i.e., if $(i,j)=\origEdge[i,j]$ holds, then Lemma~\ref{lem:local-connectivity-property} is applied to obtain a mapping of the edge $(i,j)$ together with its head $j$.
 If the edge's orientation was reversed, i.e.~iff. $(i,j) \neq \origEdge[i,j]$ holds, Lemma~\ref{lem:local-connectivity-property} can be applied again, only now a path from the mapping of \emph{the head} $i$ (according to the edge's original orientation) to some mapping of the tail $j$ is obtained.
 Lastly, the minimum mapping value $\prob$ is computed and the variables of the LP (including the allocation variables) are decreased accordingly. The formal correctness of the algorithm is proven in Lemma~\ref{lem:decomposability-mcf-trees}.
 
 \begin{restatable}{lemm}{decomposabilityOfMcfTrees}
 \label{lem:decomposability-mcf-trees}
 Given a virtual network request $\req \in \requests$, whose 
 underlying undirected graph is a tree, and a solution $(x_\req, \vec{y}_{\req}, \vec{z}_{\req}, \vec{a}_{\req})$ to the LP Formulation~\ref{alg:VNEP-IP-old}, the solution can be decomposed into convex combinations of valid mappings $\PotEmbeddings = \{(\prob,\mapping)\}_k$, 
 such that the following holds:
 \begin{itemize}
 \item The decomposition is complete, i.e., $x_{\req} = \sum_k \prob$ holds.
 \item The decomposition's resource allocations are bounded by $\vec{a}_{\req}$, i.e., 
 $a^{x,y}_{\req} \geq \sum_k \prob \cdot A(\mapping,x,y)$ holds for each resource $(x,y) \in \SR$.
 \end{itemize}
 \end{restatable}
 \begin{proof}
 Note that the mapping of each virtual node and each virtual edge is valid by 
 construction: Constraints~(\ref{alg:VNEP-old:node-embedding-non-mappable}) 
 and (\ref{alg:VNEP-old:edge-embedding-non-mappable}) enforce that a node and 
 an edge can only be mapped in a valid fashion. 
 Furthermore, as $\VGbfs$ is an arborescence, node mappings are never revoked and 
 each node of $\VG$ will eventually be mapped.
 The mapping value $\prob$ is computed as the minimum of the mapping variables $\mathcal{V}_k$ 
 used for constructing $\mapping$. Reducing the values of the mapping variables together with the load variables $\vec{a}_{\req}$, the Constraints~\ref{alg:VNEP-old:node-embedding}-\ref{alg:VNEP-old:load-edge} continue to hold.
 
 As the decomposition process continues as long as $x_{\req} > 0$ holds and in the $k$-th 
 step at least one variable's value is set to $0$, it is easy to check that (i) the algorithm terminates 
 with a complete decomposition for which $\sum_k \prob = x_{\req}$ holds and (ii) the algorithm has 
 polynomial runtime, as the number of variables for request $\req$ is bounded by $\mathcal{O}(|\VG| \cdot |\SE|)$. 
 \end{proof}

\subsection{Limitations of the MCF Formulation}
Having shown the decomposability of LP solutions for tree requests, we now show that this does not hold, if the request graphs contain \emph{cycles}. Figure~\ref{fig:non-decomp} gives an example for an LP solution of Formulation~\ref{alg:VNEP-IP-old} from which no valid mapping (that is covered) can be extracted. Concretely, considering the mapping of $i$ on $u_1$ and following the depicted extraction order, $k$ and $j$ must be mapped on $u_6$ and $u_2$, respectively. However, the mapping of $j$ on $u_2$ only allows for the mapping of $k$ on $u_3$ and no valid mapping can be extracted and we obtain:

\begin{figure}[b!]
\centering
\includegraphics[height=0.11\textheight]{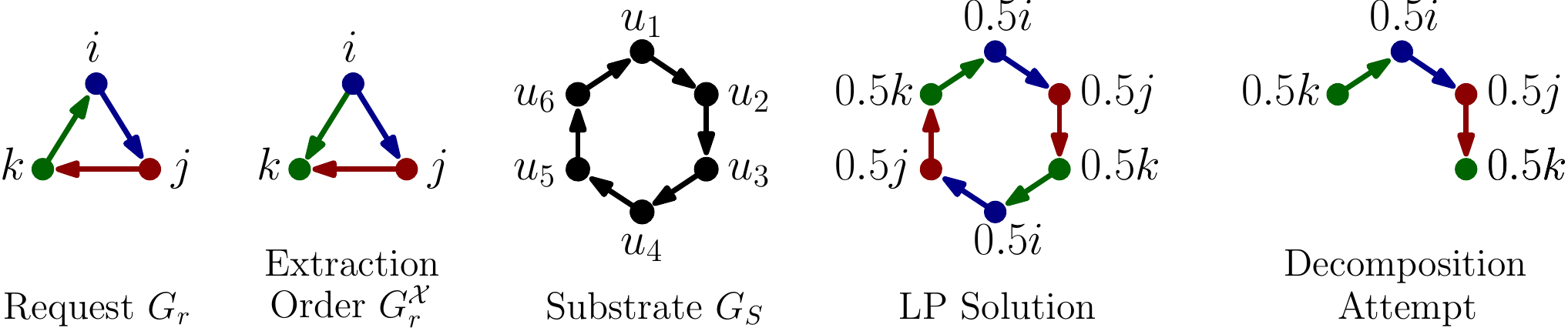}
\caption{Example showing that solutions to the LP Formulation~\ref{alg:VNEP-IP-old} cannot be decomposed into convex combinations of valid mappings. 
The LP solution with $x_\req=1$ is depicted as follows. 
Substrate nodes are annotated with virtual node mappings: $0.5i$ at node $u_1$ 
indicates $y^{u_1}_{r,i} = 1/2$. 
Substrate edge colors match the color of the virtual edges mapped to it. All virtual edges are also mapped using flow values $1/2$. The color of substrate edge $(u_1,u_2)$ therefore implies that $z^{u_1,u_2}_{r,i,j} = 1/2$ holds. }
\label{fig:non-decomp}
\end{figure}

\begin{restatable}{theorem}{thmNonDecomposabilityAndItsImplications}
Solutions to the LP Formulation~\ref{alg:VNEP-IP-old} can (in general) not be decomposed 
into convex combinations of valid mappings, if request graphs contain cycles. Accordingly, the integrality gap of the LP Formulation~\ref{alg:VNEP-IP-old} is unbounded for cyclic request graphs.
\label{thm:non-decomposability}
\end{restatable}
\begin{proof}
  Figure~\ref{fig:non-decomp} depicts an example solution to the LP Formulation~\ref{alg:VNEP-IP-old} from which \emph{not a single} valid mapping can be extracted. The validity of the depicted solution is easy to check. As virtual node $i \in \VV$ is mapped onto substrate node $u_1 \in \SV$, and $u_2 \in \SV$ is the only neighboring node with respect to the commodity $z_{\req, i,j}$ that hosts $j \in \VV$, a mapping $(\mapV, \mapE)$ with $\mapV(i)=u_1$ and $\mapV(j)=u_2$ must exist. Similarly, $\mapV(k)=u_3$ must hold. However, the flow of virtual edge $(k,i) \in \VE$ leaving $u_3 \in \SV$ only leads to $u_4 \in \SV$. Hence the virtual node $i \in \VV$ must be mapped both on $u_1$ and $u_4$. As the same argument applies when considering the mapping of $i$ onto $u_4$, no valid mapping can be extracted.
  
  We now show that the formulation exhibits an unbounded integrality gap.
  Consider the following restrictions for mapping the virtual links: $\VEloc[i,j] = \{(u_1,u_2), (u_4,u_5)\}$,
   $\VEloc[j,k] = \{(u_2,u_3), (u_5,u_6)\}$, 
  $\VEloc[k,i] = \{(u_3,u_4), (u_6,u_1)\}$.
  Note that the solution depicted in Figure~\ref{fig:non-decomp} is still feasible for the MCF LP. Considering the profit variant of the MCF formulation, the LP will attain an objective of $\Vprofit$. As on the other hand, there does not exist a valid mapping of request $\req$ on $\SG$, the optimal solution achieves a profit of $0$. Hence, the integrality gap of the profit formulation is unbounded.
  
  For the cost variant, we add an edge $(u_3,u_1)$ of arbitrarily high cost to the substrate and include this edge in the set of allowed edges for the virtual edge $(k,i) \in \VE$. Hence, there exists only a single valid mapping, which uses this edge $(u_3,u_1)$ while the MCF formulation might still use the LP solution depicted in Figure~\ref{fig:non-decomp}. Hence, as the cost of the edge $(u_3,u_1)$ can be arbitrarily high, the integrality gap is unbounded.
  \end{proof}

\section{Novel Decomposable LP Formulation}\label{sec:decomp}

\DeclareDocumentCommand{\VEbfsLabelsBag}{O{b} O{i} O{\req}}{\ensuremath{\mathcal{L}^{\extractionOrderCharacter,#1}_{#3,#2}}}

\DeclareDocumentCommand{\MappingSpace}{d[]}{\ensuremath{\mathcal{M}(#1)}}


\DeclareDocumentCommand{\VGe}{O{e} O{\req}}{\ensuremath{G_{\req,e}}}
\DeclareDocumentCommand{\VVe}{O{e} O{\req}}{\ensuremath{V_{\req,e}}}
\DeclareDocumentCommand{\VEe}{O{e} O{\req}}{\ensuremath{E_{\req,e}}}

\DeclareDocumentCommand{\mapVedge}{O{e}}{\ensuremath{m^{\mathcal{L}}_{#1}}}
\DeclareDocumentCommand{\mapVbag}{O{a}}{\ensuremath{m^{\mathcal{L}}_{#1}}}
\DeclareDocumentCommand{\mapVinter}{O{b} O{e}}{\ensuremath{m^{\mathcal{L}}_{#1 \cap #2}}}

\DeclareDocumentCommand{\VEbfsLabelsEdgeCapBag}{O{b} O{e}}{\ensuremath{\mathcal{L}^{\extractionOrderCharacter}_{#1 \cap #2}}}

\label{sec:novel-formulation}

In this section we present our novel LP formulation to solve the fractional VNEP and is the basis for our randomized rounding approximation algorithms for the VNEP. We first present the high-level idea of our formulation and introduce crucial concepts as extraction confluences and the extraction width. After introducing further notation, we formally present the LP formulation and show the decomposability of its solutions.

\subsection{Idea and Definitions}
\label{sec:novel-formulation:idea}
We shortly outline the key idea of our formulation by analyzing the shortcomings of the MCF formulation. Considering the example of Figure~\ref{fig:non-decomp}, we observe that there exist two virtual paths towards $k$ in $\VGbfs$, namely $\langle (i,k) \rangle$ and $\langle (i,j), (j,k) \rangle$. We refer to the combination of two such paths in $\VGbfs$ leading from a common virtual node to another common node as an \emph{extraction confluence}:
\begin{definition}[Extraction Confluence $\VEbfsAC$]
Given an extraction order $\VGbfs$, an \emph{extraction confluence} $\VEbfsAC = \VEbfsACL \sqcup \VEbfsACR$ connects $i \in \VV$ to $j \in \VV$ using two otherwise node-disjoint paths $\VEbfsACL, \VEbfsACR \subseteq \VEbfs$. We refer to $i$ as the source and to $j$ as the target of the confluence $\VEbfsAC$.
\end{definition}
According to the connectivity property of the MCF formulation (cf.~Lemma~\ref{lem:local-connectivity-property}), (partial) mappings can always be extended, but the disjoint paths of a confluence might lead to \emph{different} mappings of the confluence's target as depicted in Figure~\ref{fig:non-decomp}. However, this \emph{divergence} is only possible when the confluence's target can be mapped on multiple locations and is not fixed.

\begin{wrapfigure}{R}{0.37\textwidth}
\vspace{-12pt}
\includegraphics[width=0.35\textwidth]{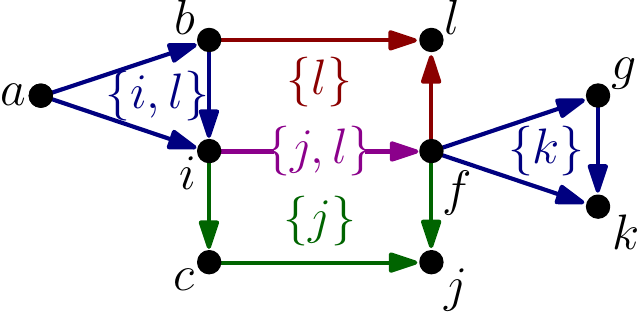}
\vspace{-6pt}
\caption{Exemplary labeled $\VGbfs$.}
\label{fig:labeling-example}
\vspace{-12pt}
\end{wrapfigure}

We use this as follows. Considering a confluence $\VEbfsAC$, our LP formulation considers multiple copies of the MCF formulation \emph{for each potential mapping location of the confluence's target}. In each of these copies, the mapping of the confluence's target is fixed to a \emph{specific} substrate node. To generalize this idea to multiple confluences, we label edges with confluence targets as follows.

\begin{definition}[Extraction Edge Labels] 
\label{def:edge-labels}
We introduce edge labels $\VEbfsLabels \subseteq \VV$ for $e \in \VEbfs$ as follows. The extraction order edge $e$ is labeled with node $j$, i.e., $j \in \VEbfsLabels$ holds, iff. a confluence $\VEbfsAC$ with target $j$ exists that contains $e$.
We also label the edges in their original orientation accordingly: for edge $e \in \VE$ we set $\VEbfsLabelsOrig \triangleq \VEbfsLabels[e']$ with $e'=\extractionEdge$.
\end{definition}

The edge labels will be used in our novel LP formulation to instantiate copies of the MCF formulation. Additionally, we introduce \emph{confluence edge bags} which partition outgoing edges.

\begin{definition}[Confluence Edge Bags] 
\label{def:edge-bags}
Given an extraction order $\VGbfs$, the outgoing edges $\deltaPlusA$ of each node $i \in \VV$ are partitioned into a set of edge bags $\VEbfsBags = \{\VEbfsBag\}_\VEbfsBagIterator$, such that two edges $e_1,e_n \in \deltaPlusA$ are placed in the same bag $\VEbfsBag$, iff.~there exists a series of edges $e_2,e_3\dots,e_{n-1} \in \deltaPlusA$ such that $\VEbfsLabels[e_l] \cap \VEbfsLabels[e_{l+1}] \neq \emptyset$ holds for $l \in \{1,\ldots,n-1\}$. 

We denote by $\VEbfsLabelsBag = \bigcup_{e \in \VEbfsBag} \VEbfsLabels$ the union of labels contained in a bag $\VEbfsBag \in \VEbfsBags$ and by $\VEbfsLabelsEdgeCapBag = \VEbfsLabels \cap \VEbfsLabelsBag$ the intersection of labels of the bag $\VEbfsBag$ and the edge $e \in \VEbfs$.
\end{definition}

The size of our formulation will be proven to be exponential in the maximal number of labels contained in any edge bag, and we define the notion of \emph{extraction width} accordingly:

\begin{definition}[Extraction Width] 
The width $\ewX(\VGbfs)$ of a given extraction order $\VGbfs$ is the maximal number of labels contained in an edge bag plus one: $\ewX(\VGbfs) = 1 + \max_{i \in \VV, \VEbfsBag \in \VEbfsBags} |\VEbfsLabelsBag|$.
Denoting by $\extractionOrderCharacter(\VG)$ the set of all extraction orders of a graph $\VG$, the extraction width of an arbitrary graph $\VG$ is the minimum width of any extraction order: $\ew(\VG) = \min_{\VGbfs \in \extractionOrderCharacter(\VG)} \ewX(\VGbfs)$.
\end{definition}

Figure~\ref{fig:labeling-example} depicts an example extraction order containing 5 confluences, which can be uniquely identified by their sources and targets: $\VEbfsAC[a][i]$, $\VEbfsAC[i][j]$, $\VEbfsAC[a][l]$, $\VEbfsAC[b][l]$, $\VEbfsAC[f][k]$ with e.g. $\VEbfsAC[b][l] = \{(b,i),(i,f),(f,l)\} \sqcup \{(b,d),(d,l)\}$.  
 According to Definition~\ref{def:edge-bags}, the edge bags of node $f$ are  $\VEbfsBags[f] = \{ \VEbfsBag[1][f]=\{(f,j)\}, \VEbfsBag[2][f]=\{(f,g),(f,k)\}, \VEbfsBag[3][f]=\{(f,l)\} \}$ with the corresponding label sets being $\VEbfsLabelsBag[1][f]=\{j\}$, $\VEbfsLabelsBag[2][f]=\{k\}$, and $\VEbfsLabelsBag[3][f]=\{l\}$. For node $i$, only a single edge bag $\VEbfsBags[i] = \{ \VEbfsBag[1][i]=\{(i,c),(i,f)\}\}$with label set $\VEbfsLabelsBag[1][i] = \{j,l\}$ exists. 

\subsection{Structure of Edge Labels}

In the following, we study the structure of extraction confluences and of the edge labels. We employ the following notation for indicating edges being reachable from and/or by nodes in the extraction order.

\begin{definition}[Reachable Edge Sets] 
Given an extraction order $\VGbfs$, we denote by $\VEbfsSuc, \VEbfsPre \subseteq \VEbfs$ the set of edges which can be reached from $i \in \VV$ and which may lead to $j \in \VV$. We denote by $\VEbfsInter = \VEbfsSuc \cap \VEbfsPre $ the edges lying on a path from $i$ to $j$.
\end{definition}

The following lemma forms the basis for efficiently computing edge labels.

\begin{restatable}{lemma}{lemEdgeLabelingConditions}
\label{lem:labels:edge-labeling-conditions}
Edge $e \in \VEbfs$ is labeled with $j\in \VV$ iff.~ there exists a node $i \in \VV$, such that (i)~$e$~lies on a path from $i$ to $j$, i.e. $e \in \VEbfsInter$, and (ii) a confluence $\VEbfsAC$ from $i$ to $j$ exists.
\end{restatable}
\begin{proof}
It is easy to see that the above two conditions are necessary.
Clearly, if the first condition does not hold for some node $i \in \VV$, then there cannot exist a confluence from $i$ to $j$ covering the edge $e$. Secondly, if there does not exist any confluence between $i$ and $j$, then there cannot exist a confluence from $i$ to $j$ covering $e$.

We now show that these conditions are also sufficient.
First, note that any path from $i$ to $j$ must be contained in $\VEbfsInter$. Let $e \in \VEbfsInter$ denote any edge for which the above conditions hold. We show that edge $e$ lies on a confluence with target $j$. By the second condition, there exist two node-disjoint paths $P^1_{i,j}, P^2_{i,j} \subseteq \VEbfsInter$ from $i$ to $j$. Now, if $e$ lies on either of these paths, then $P^1_{i,j} \sqcup P^2_{i,j}$ already constitutes a confluence. Hence, assume that $e$ does not lie on either of these paths. Let $e=(k,l)$, i.e. $k$ is the tail and $l$ the head. Furthermore, let $P_{i,k} \subseteq \VEbfsInter$ denote any path from $i$ to $k$ and denote by $P_{l,j} \subseteq \VEbfsInter$  any path from $l$ to $j$. Let $P_{i, e, j}$ denote the path obtained from joining $P_{i,k}$, $e=(k,l)$, and $P_{l,j}$.

If $P_{i,e,j}$ only intersects with $P^1_{i,j}$ (or $P^2_{i,j}$), then $P_{i,e,j}$ together with $P^2_{i,j}$ (or $P^1_{i,j}$) constitutes a confluence towards $j$ which covers $e$, proving our claim. Hence, assume that $P_{i,e,j}$ intersects with both paths. Let $k'$ be the last node on path $P_{i,k}$ which also lies on $P^1_{i,j}$ or $P^2_{i,j}$ and let $l'$ denote the first node on $P_{l,j}$ which also lies on $P^1_{i,j}$ or $P^2_{i,j}$. Assume now  w.l.o.g.~that both $k'$ and $l'$ lie on path $P^1_{i,j}$, then the subpath of $P^1_{i,j}$ from $k'$ to $l'$ can be substituted with the subpath from $k'$ to $l'$ of $P_{i,e,j}$, yielding the confluence depicted on the left of Figure~\ref{fig:labeling_proof}. On the other hand, if $k'$ lies on $P^1_{i,j}$ while $l'$ lies on path $P^2_{i,j}$, then there exists a confluence from $k'$ to $j$ covering the edge $e$: using the suffix of path $P^1_{i,j}$ starting at node $k'$ as first path and the subpath of $P_{i,e,k}$ from $k'$ to $l'$ together with the suffix of $P^2_{i,j}$ starting at $l'$, a confluence is found that covers $e$. By construction, as the nodes of path $P_{i,e,j}$ between $k'$ and $l'$ do neither lie on $P^1_{i,j}$ nor $P^2_{i,j}$, the paths of the constructed confluence are disjoint (see Figure~\ref{fig:labeling_proof} (center) for a visualization).
Hence, the two conditions stated in the lemma are also sufficient to decide whether an edge $e \in \VEbfs$ is covered by a confluence towards $j \in \VV$ holds.
\end{proof}

\begin{figure}[tb!]
\hspace{32pt}
\includegraphics[height=0.15\textheight]{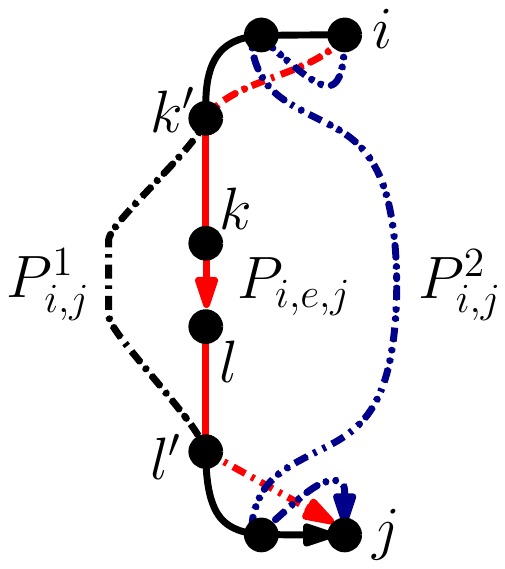}
\hfill
\includegraphics[height=0.15\textheight]{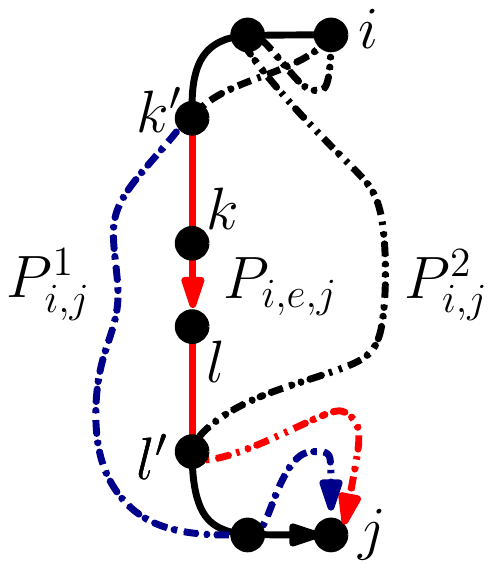}
\hfill
\includegraphics[height=0.15\textheight]{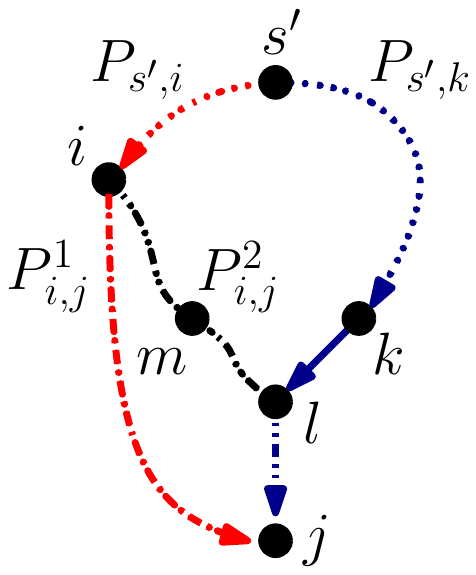}
\hspace{32pt}
\caption{Visualizations of the constructions used in the proofs of Lemma~\ref{lem:labels:edge-labeling-conditions} (left and center) and Lemma~\ref{lem:labels:in-edge-labeled-same} (right) to show that an edge $(k,l)$ is covered by a confluence. The confluence path $P^1_{i,j}$ is dashed with a single dot and the confluence path $P^2_{i,j}$ is dashed with two dots. The constructed confluences consist of the highlighted paths in red and blue. \newline
Lemma~\ref{lem:labels:edge-labeling-conditions} (left and center): Construction of a confluence when $P^1_{i,j}$ intersects with $k'$ and $l'$ (left). Construction of a confluence when $k'$ lies on $P^1_{i,j}$, while $l'$ lies on $P^2_{i,j}$ (center). \newline
Lemma~\ref{lem:labels:in-edge-labeled-same} (right): As there must exist a node $r'$ reaching both $i$ and $m$, a confluence with target $j$ is constructed covering the edge $(m,l)$.}
\label{fig:labeling_proof}
\end{figure}

Based on Lemma~\ref{lem:labels:edge-labeling-conditions}, the edge labels can be computed in polynomial-time. Concretely we apply Menger's theorem~\cite{menger1927allgemeinen} to decide for any combination of virtual nodes $i,j \in \VV$, whether two disjoint paths exist from $i$ to $j$. If this is the case, then all edges lying in $\VEbfsInter$ must be labeled with $j$ and we obtain:

\begin{restatable}{lemma}{lemEdgeLabelComputationViaMenger}
\label{lem:labels:edge-label-computation-via-menger}
The edge labels $\VEbfsLabels$ can be computed in time $\mathcal{O}(|\VV|^3 \cdot |\VE|$).
\end{restatable}
\begin{proof}
We argue that the conditions of Lemma~\ref{lem:labels:edge-labeling-conditions} can be checked in polynomial time. For each potential target node $j \in \VV$ and each source node $i \in \VV$, we check whether two node-disjoint paths exist from $i$ to $j$ by applying Menger's theorem~\cite{menger1927allgemeinen}: for each node $k$ lying on a path from $i$ to $j$, we decide whether $j$ is still reachable from $i$ when $k$ is removed. If this is true for each intermediate node, then by Menger's theorem, there exist at least two node-disjoint paths from $i$ to $j$ and hence there must exist a confluence $\VEbfsAC$. Hence, given the existence of a confluence, all edges in $\VEbfsInter$ are labeled by $j$. At most $\mathcal{O}(|\VV|^2)$ many node pairs need to be considered and the check whether two node-disjoint paths exist can be implemented in time $\mathcal{O}(|\VV|\cdot |\VE|)$. Hence, the overall runtime to compute all labels is bounded by $\mathcal{O}(|\VV|^3\cdot|\VE|)$.
\end{proof}

The two following lemmas state important structural properties for edge labels, namely that incoming edges are always labeled the same and that each label has a unique source.

\begin{restatable}{lemma}{lemInEdgeLabelsSame}
\label{lem:labels:in-edge-labeled-same}
$\VEbfsLabels[e] = \VEbfsLabels[e']$ holds for any pair of 
incoming edges $e,e' \in \deltaMinusA[l]$ of node $l \in \VV$.
\end{restatable}
\begin{proof}
Assume for the sake of contradiction that an edge $e=(m,l)$ is labeled with $j$, i.e.~$j \in \VEbfsLabels$, and that some other incoming edge $e'=(k,l)$ is not labeled with $j$, i.e.~$j \notin \VEbfsLabels[e']$.
As the edge $e$ is labeled with $j$, there must exist some confluence $\VEbfsAC$ covering $e$. 
As the edge $e'=(k,l)$ is not labeled with $j$, we obtain from Lemma~\ref{lem:labels:edge-labeling-conditions} that the edge $e'$ is not reachable from $i$, i.e., $e' \notin \VEbfsInter$ holds. As both $i$ and $k$ are reachable from the root $\VVroot$ of $\VGbfs$, there must exist paths $P_{\VVroot,i}$ and $P_{\VVroot,k}$ leading from the root $\VVroot$ to $i$ and $k$, respectively. Let $s'$ denote the last node that lies on both of these paths. The subpaths $P_{s',i}$ and $P_{s',k}$ of $P_{\VVroot,i}$ and $P_{\VVroot,k}$ starting at $s'$ do not use any of the edges in $\VEbfsInter$. Hence, joining $P_{s',i}$ with $P^1_{i,j}$ and joining $P_{s',k}$ with $e'=(k,l)$ and the subpath of $P^2_{i,j}$ beginning at node $l$, a confluence is constructed that covers edge $e'$ (see Figure~\ref{fig:labeling_proof} (right) for an visualization of the construction). Therefore, also $e'$ must be labeled with $j$, yielding a contradiction to our assumption that $e'$ was not labeled by $j$ and all incoming edges must be labeled the same.
\end{proof}

Lastly, the following lemma shows that any label is introduced only once.

\begin{restatable}{lemma}{lemUniqueRootForLabel}
\label{lem:labels:unique-root-for-label}
For each label $j \in \VV$ there exists a unique root node $s_j \in \VV$, such that:
\begin{enumerate}
\item Any edge $e \in \VEbfs$ being labeled with $j \in \VEbfsLabels$ is contained in $\VEbfsInter[s_j][j]$.
\item Any path from $\VVroot$ (the root of the extraction order)  to $j$ passes through $s_j$.
\end{enumerate}
\end{restatable}
\begin{proof}
Consider two nodes $i,i' \in \VV$ being the sources of confluences $\VEbfsAC[i][j]$ and $\VEbfsAC[i'][j]$ towards $j$. Assume that neither $i$ occurs in $\VEbfsInter[i'][j]$ nor that $i'$ occurs in $\VEbfsInter[i][j]$. As the graph $\VGbfs$ is rooted, there must exist a node $s'$ reaching both $i$ and $i'$ and spawning a confluence towards $j$. Furthermore, $(\VEbfsInter[i][j] \cup \VEbfsInter[i'][j]) \subseteq \VEbfsInter[s'][j]$ holds in this case.
Hence, for any pair of nodes $i,i'$ being sources of confluences towards $j$, either one of the nodes is reachable from the other, or there exists another node $s' \in \VV$ such that $\VEbfsInter[i][j]$ and $\VEbfsInter[i'][j]$ are contained in $\VEbfsInter[s'][j]$. Clearly, as either $i$ dominates $i'$ or $i'$ dominates $i$, or there exists some other node $s'$ dominating both, there must exist a single unique root node $s_j \in \VV$ such that all edges labeled with $j$ lie in $\VEbfsInter[s_j][j]$.

The second claim is immediate: if there was to exist some path from the root $\VVroot$ to an edge being labeled with $j$ without passing through the unique root  $s_j \in \VV$, then there must exist a confluence $\VEbfsAC[s'][j]$ starting at some other node $s' \in \VV$, such that $s'$ reaches $s_j$ but $s_j$ does not reach $s'$. Hence, by our above observation $s'$ dominates $s_j$ and $s_j$ cannot be the unique root. Thus, all paths from the root must pass through $s_j$ on their way to $j$.
\end{proof}

\subsection{Novel Decomposable LP Formulation}
\label{sec:novel-formulation:novel-formulation}
Our novel Linear Programming Formulation~\ref{IP:novel-AC} is based on the idea to decide the mapping locations of confluence targets a priori. We do so by considering copies or sub LPs of the MCF formulation (see Constraint~\ref{LP:novel:subformulations}) and we employ the following notation. For an edge $e=(i,j) \in \VE$, we denote by $\VGe = (\VVe, \VEe)$ with $\VVe = \{i,j\}$ and $\VEe = \{e\}$ the subgraph of $\VG$ containing only edge $e$. Variables of sub LPs are named as before, but are now additionally indexed: $\subLP[\alpha][\beta]$ denotes the variable $\alpha$ in the copy identified by $\beta$. To denote the combinations of mapping possibilities of labels, we employ $\MappingSpace[X]$ to denote the function space from the set $X$ to $\SV$, i.e., $\MappingSpace[X] = [X \to \SV]$. Accordingly, considering an edge $e \in \VE$ of request $\req \in \requests$ being labeled by $\VEbfsLabelsOrig$, we instantiate one copy of the MCF formulation per edge label mapping $\mapVedge \in \MappingSpace[\VEbfsLabelsOrig]$ (cf.~Constraint~\ref{LP:novel:subformulations}).
For better readability, we write $\restrict[f][Z]:Z\to Y$ to denote $f_{|Z}$, i.e., the (standard) restriction of the function $f: X \to Y$ on the subset $Z \subseteq X$. Hence, $\restrict[f][Z](z) = f(z)$ holds for $z \in Z$.

To link the LP copies, we employ two types of \emph{global} node mapping variables. We use the (global) $y^u_{\req,i}$ variables already presented in Formulation~\ref{alg:VNEP-IP-old} as well as node mapping variables $\gamma^u_{\req,i,b,a} \in [0,1]$ for \emph{edge bags} $\VEbfsBag \in \VEbfsBags$, each mapping $\mapVbag \in \MappingSpace[\VEbfsLabelsBag]$ of the labels contained in the respective edge bag, and the mapping locations $u \in \VVloc$. The classic variables $y^u_{\req,i}$ are used for coupling the embedding variable $x_{\req}$ and the sub-LP node mapping variables (see Constraints~\ref{LP:novel:node-embedding} and \ref{LP:novel:node-to-sub-node-mapping}) as well as for computing the node load (see Constraint~\ref{LP:novel:node-load}). 
The node mapping variables for edge bags  $\gamma^u_{\req,i,b,a}$ are defined for all mappings of their label set  $\mapVbag \in \MappingSpace[\VEbfsLabelsBag]$. As $\VEbfsLabels \subseteq \VEbfsLabelsBag$ holds for all edges $e \in \VEbfsBag$, the node mapping variables of an edge bag directly induce node mappings for all edges contained in the respective bag (see Constraint~\ref{LP:novel:gamma-to-outgoing-edges}).

\begin{figure}[tb]

 {
 \LinesNotNumbered
  \renewcommand{\arraystretch}{0}
 
 \removelatexerror

  \begin{IPFormulation}{H}
  
  \popline
 
  \SetAlgorithmName{Formulation}{}{{}}

  \newcommand{\spaceIt}{\qquad\quad\quad}
  \newcommand{\miniSpace}{\hspace{1.5pt}}

  \begin{tabular}{FRLQB}
    \multicolumn{5}{r}{\parbox{0.92\textwidth}{~}} \\ 

   \multicolumn{3}{L}{\textnormal{  \hspace{-10pt}(\ref{alg:VNEP-old:node-embedding}) -    (\ref{alg:VNEP-old:load-edge}) for $\VGe$ on variables } \subLP[(x_{\req},\vec{y}_{\req},\vec{z}_{\req},\vec{a}_{\req})][e,\mapVedge]} ~~& \forall \req \in \requests, e \in \VE, \mapVedge \in \MappingSpace[\VEbfsLabelsOrig]   & \tagIt{LP:novel:subformulations}
   \\[6pt]
   
  x_{\req} & ~=~ & \sum \limits_{u \in \VVloc} y^u_{\req, i} & \forall \req \in \requests &  \tagIt{LP:novel:node-embedding} \\
   
   y^u_{\req,i} & ~=~ & 	\sum_{\mapVedge \in \MappingSpace[\VEbfsLabelsOrig]} \hspace{-12pt} \subLP[y^u_{\req,i}][e,\mapVedge] & \forall \req \in \requests,  i \in \VV,  u \in \VVloc, e \hspace{-1pt} \in \hspace{-1pt}\VE \hspace{-1pt}: \hspace{-1pt}i \in \VVe &  \tagIt{LP:novel:node-to-sub-node-mapping} \\
       
   \subLP[y^u_{\req,i}][\origEdge,\mapVedge]    & ~=~ & \sum_{\begin{subarray}{c}
            \mapVbag \in \MappingSpace[\VEbfsLabelsBag]: \\
            \restrict[\mapVbag][\VEbfsLabelsEdgeCapBag] = \mapVedge
            \end{subarray}} \gamma^u_{\req,i,b,a}  & \hspace{-5pt}\begin{array}{r}
     \forall \req \in \requests, i \in \VV, u \in \VVloc, \VEbfsBag \in \VEbfsBags, \\ e \in \VEbfsBag, \mapVedge \in \MappingSpace[\VEbfsLabels]
     \end{array} & \tagIt{LP:novel:gamma-to-outgoing-edges}\\

   \sum_{\begin{subarray}{c}
   \mapVedge \in \MappingSpace[\VEbfsLabels]: \\
   \restrict[\mapVedge][\VEbfsLabelsEdgeCapBag] =  \mapVinter
   \end{subarray}} \hspace{-16pt} \subLP[y^u_{\req,i}][\origEdge,\mapVedge] & ~=~ & \sum_{\begin{subarray}{c}
      \mapVbag \in \MappingSpace[\VEbfsLabelsBag]: \\
      \restrict[\mapVbag][\VEbfsLabelsEdgeCapBag] =  \mapVinter
      \end{subarray}} \hspace{-16pt} \gamma^u_{\req,i,b,a}  &  \hspace{-5pt}\begin{array}{r}
            \forall \req \in \requests, i \in \VV, u \in \VVloc, e \in \deltaMinusA,\\ \VEbfsBag \in \VEbfsBags, \mapVinter \in \MappingSpace[\VEbfsLabelsEdgeCapBag]
            \end{array}  & \tagIt{LP:novel:incoming-edges-to-gamma-variables}\\

     \subLP[y^u_{\req,i}][\origEdge,\mapVedge] & ~=~ & 0  & \hspace{-5pt}\begin{array}{r}
              \forall \req \in \requests, i \in \VV, e \in \deltaMinusA: i \in \VEbfsLabels,\\ 
              \mapVedge \in \MappingSpace[\VEbfsLabels], u \in \VVloc \setminus \{\mapVedge(i)\}
              \end{array}  & \tagIt{LP:novel:forbidding-nodes-in-sub-lps}\\
	
a^{\type,u}_{\req}  & ~=~  & 	 	\sum \limits_{i \in \VV, \Vtype(i) = \type}  \hspace{-12pt}  \Vcap(i) \cdot y^u_{\req,i} & \forall \req \in \requests, (\type,u) \in  \SRV &  \tagIt{LP:novel:node-load}\\
   	
	a^{u,v}_{\req} & ~=~ & \sum_{\begin{subarray}{c}
e \in \VE\\ \mapVedge \in \MappingSpace[\VEbfsLabelsOrig]
	\end{subarray}} \hspace{-12pt} \subLP[a^{u,v}_{\req}][e,\mapVedge]    \quad & \forall \req \in \requests, (u,v) \in  \SE&  \tagIt{LP:novel:edge-load}\\
	
	\sum \limits_{\req \in \requests} a^{x,y}_{\req}  & ~\leq~ & \Scap(x,y) & \forall (x,y) \in  \SR &  \tagIt{LP:novel:capacities}  \\[12pt]


\multicolumn{4}{C}{
\hspace{-10pt}
\begin{array}{c}
x_{\req} \in [0,1],~\forall \req \in \requests; \hspace{12pt} y^u_{\req,i} \in [0,1],~\forall \req \in \requests, i \in \VV, u \in \VVloc; \hspace{12pt} a^{x,y}_{\req} \geq 0,~\forall \req \in \requests, (x,y) \in \SR \\[6pt]
\gamma^u_{\req,i,b,a} \in [0,1],~\forall \req \in \requests, i \in \VV, u \in \VVloc, \VEbfsBag \in \VEbfsBags, \mapVbag \in \MappingSpace[\VEbfsLabelsBag]
\end{array}

} & \tagIt{LP:novel:variables}
  \end{tabular}
  \caption{Novel Decomposable Base Formulation for the VNEP}
  \label{IP:novel-AC}
  \end{IPFormulation}
  }
\end{figure}

\begin{figure}[tb]
\centering
\includegraphics[width=0.91\textwidth]{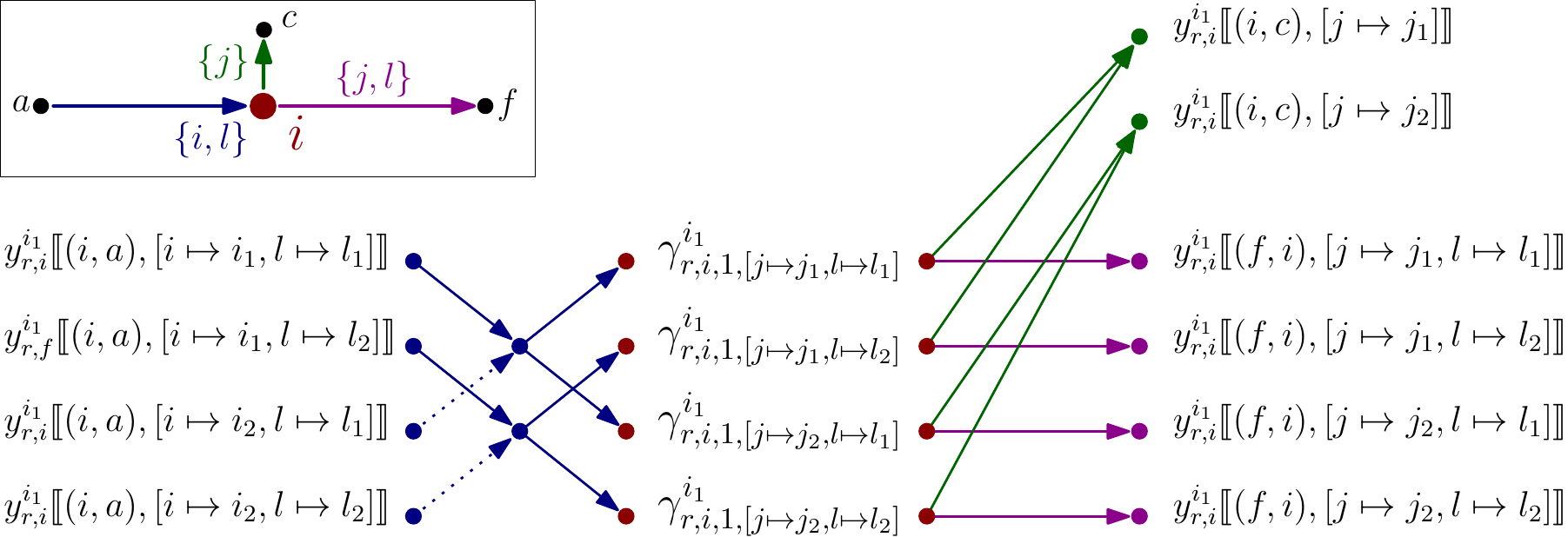}
\caption{Visualization of the relation of the different node mapping variables for the example of Figure~\ref{fig:labeling-example} under the assumption that the virtual nodes $i,j,k \in \VV$ can be mapped only on $i_k, j_k, l_k \in \SV$ for $k \in \{1,2\}$, respectively. Depicted are only the variables relating to the mapping of $i$ to $i_1$. To highlight that the LP copies are created upon the original orientation of edges, we assume that $(i,a), (f,i) \in \VE$ holds and that these were reversed for $\VGbfs$ depicted in Figure~\ref{fig:labeling-example} on Page~\pageref{fig:labeling-example}. 
\newline
The edges and nodes are to be read as `flows' for which flow preservation holds. The directions of the edges shall help the reader to follow how the node mapping variables of `incoming' edges trigger the node mapping variables of `outgoing' edges. The connections on the left are due to Constraint~\ref{LP:novel:incoming-edges-to-gamma-variables} and the connections on the right are due to Constraint~\ref{LP:novel:gamma-to-outgoing-edges}. Note that the dashed edges on the left will be 0 due to Constraint~\ref{LP:novel:forbidding-nodes-in-sub-lps}: in the index of the respective sub LP, the virtual node $i$ is mapped onto $i_2$ and hence the respective dashed variables are set to $0$.
}
\label{fig:lp-formulation-node-variables}
\end{figure}

In the following, we argue how `flows' are induced and accordingly how solutions to the formulation can be decomposed. Figure~\ref{fig:lp-formulation-node-variables} visualizes the workings of Constraints~\ref{LP:novel:gamma-to-outgoing-edges} to \ref{LP:novel:forbidding-nodes-in-sub-lps}.
Assuming that $x_{\req} > 0$ holds, then by Constraint~\ref{LP:novel:node-embedding} there will exist a substrate node $u \in \VVloc[\VVroot]$ onto which the root $\VVroot$ is placed, i.e. $y^u_{\req,\VVroot} > 0$ holds. Constraint~\ref{LP:novel:node-to-sub-node-mapping} distributes the quantity $y^u_{\req,\VVroot}$ over the sub LP node mapping variables while Constraint~\ref{LP:novel:gamma-to-outgoing-edges} ensures that these node mapping variables agree with each other. Due to the validity of the MCF Formulation~\ref{alg:VNEP-IP-old}, by setting the node mapping variable for one of the endpoints of the edge graph $\VGe$, the node mapping variables of the other endpoint of $\VGe$ must be set accordingly. On the other hand, Constraint~\ref{LP:novel:incoming-edges-to-gamma-variables} ensures that any incoming edge (according to the extraction order) agrees with the respective node bag variables and hence force the further distribution of `flows'. 
The correctness of the formulation then mostly follows from the following observations: (i) Based on the acylicity of the extraction order and the fact that all nodes can be reached from the root $\VVroot$, `flow' is distributed throughout the whole request graph. (ii) A novel edge label $j$ is introduced only exactly once according to Lemma~\ref{lem:labels:unique-root-for-label}, namely at the root $s_j \in \VV$ and hence only at node $s_j$ the a priori mapping of $j$ is fixed. (iii) For any confluence $\VEbfsAC[k][i]$ with target $i$ all incoming edges of $i$ are itself labeled by $i$ (cf.~Lemma~\ref{lem:labels:in-edge-labeled-same}). Accordingly, considering a mapping $\mapVedge \in \MappingSpace[\VEbfsLabelsOrig]$ of an incoming edge $e \in \deltaMinusA[i]$, Constraint~\ref{LP:novel:forbidding-nodes-in-sub-lps} explicitly forbids node placements of $i$ to nodes $\VVloc \setminus \{\mapVedge(i)\}$ in the respective sub LPs. Hence, incoming edges of node $i$ labeled by $\mapVedge$ can only map $i$ to $\mapVedge(i) \in \VVloc$.

\vspace{-6pt}
\begin{restatable}{theorem}{thmSizeOfFormulation}
\label{thm:size-of-formulation-is-exponential-in-ex-width}
Considering specific extraction orders $\VGbfs$ for each request $\req \in \requests$, the size of the novel decomposable Formulation~\ref{IP:novel-AC} is bounded by $\mathcal{O}(\sum_{\req \in \requests} |\SG|^{\ewX(\VGbfs)} \cdot |\VG|)$.
\end{restatable}
\begin{proof}
Consider a single request $\req \in \requests$ and a fixed extraction order $\VGbfs$. There are at most $\ewX(\VGbfs)-1$ many sub LPs for each edge $e \in \VEbfs$, as $|\VEbfsLabels| \leq \ewX(\VGbfs) - 1$ holds by definition.
Otherwise, the formulation's size is dominated by the node bag mapping variables  $\gamma^u_{\req,i,b,a}$  and the respective Constraints~\ref{LP:novel:gamma-to-outgoing-edges} - 
\ref{LP:novel:forbidding-nodes-in-sub-lps}. Since the bags partition the outgoing edges and encode all potential mappings of the respective label set $\VEbfsLabelsBag$ while including the mapping location of the respective virtual node $i \in \VV$, the size of the respective formulation parts is bounded by $\mathcal{O}(|\VV| \cdot |\SV|^{\ewX(\VGbfs)} )$. The result is obtained by summing over the requests.
\end{proof}

\begin{figure}[p!]
\centering
\scalebox{0.92}{
\begin{minipage}{1.07\columnwidth}

\removelatexerror

\begin{algorithm*}[H]

\SetKwInOut{Input}{Input}\SetKwInOut{Output}{Output}
\SetKwFunction{ProcessPath}{ProcessPath}{}{}
\SetKwFunction{reverse}{reverse}{}{}
\SetKwFunction{LP}{LP}
\SetKwFunction{LP}{LP}

\newcommand{\SET}{\textbf{set~}}
\newcommand{\ADD}{\textbf{add~}}
\newcommand{\EACH}{\textbf{each~}}
\newcommand{\DEFINE}{\textbf{define~}}
\newcommand{\AND}{\textbf{and~}}
\newcommand{\LET}{\textbf{let~}}
\newcommand{\WITH}{\textbf{with~}}
\newcommand{\COMPUTE}{\textbf{compute~}}
\newcommand{\FIND}{\textbf{find~}}
\newcommand{\CHOOSE}{\textbf{choose~}}
\newcommand{\DECOMPOSE}{\textbf{decompose~}}
\newcommand{\FORALL}{\textbf{for all~}}
\newcommand{\OBTAIN}{\textbf{obtain~}}
\newcommand{\WITHPROBABILITY}{\textbf{with probability~}}

\Input{Request $\req \in \requests$, extraction order $\VGbfs$, solution to Formulation~\ref{IP:novel-AC}}
\Output{Convex combination~$\PotEmbeddings = \{\decomp = (\prob,\mapping)\}_k$ of valid mappings}
\BlankLine


	\SET $\PotEmbeddings  \gets \emptyset$ \AND $k \gets 1$\\
	\While{$x_{\req} > 0$ }
	{
		
		\SET $\mapping = (\mapV,\mapE)~\gets (\emptyset,\emptyset)$ \label{alg:sc-decomposition:init-maps}\\
		
		\SET $\Queue \gets \{\VVroot \}$ \\
		
		\CHOOSE $u \in \VVloc[\VVroot]$ \WITH $y^u_{\req,\VVroot} > 0$ \AND \SET $\mapV(\VVroot)~\gets u$\\
		
		\While{$|\Queue| > 0$}{	\label{alg:decomposition:begin-while-q}
			\CHOOSE $i \in \Queue$ \AND \SET$\Queue \gets \Queue \setminus \{i\}$\\
			\ForEach{$\VEbfsBag \in \VEbfsBags$}{
				\LET $M^V = (\mapV)^{-1}(\SV)$ denote the already mapped nodes\\
				\CHOOSE $\mapVbag \in \MappingSpace[\VEbfsLabelsBag]$, s.t. $\gamma^{\mapV(i)}_{\req,i,b,a} > 0$ \AND $\restrict[\mapVbag][\VEbfsLabelsBag \cap M^V] = \restrict[\mapV][\VEbfsLabelsBag \cap M^V]$ \\
				\SET $\mapV(j) \gets \mapVbag(j)$ \FORALL $j \in \VEbfsLabelsBag \setminus M^V$\\
				\ForEach{$e=(i,j) \in \VEbfsBag$}{
					\eIf{$(i,j) = \origEdge[i,j]$}{
						\COMPUTE path ${P}^{u,v}_{\req,i,j}$ from $\mapV(i)=u$ to  $v \in \VVloc[j]$ according to Lemma~\ref{lem:local-connectivity-property}\\
					\pushline\pushline  \nonl s.t. 
					$\begin{array}{rl}
					\subLP[y^v_{\req,j}][(i,j),\restrict[\mapV][\VEbfsLabelsOrig]] > 0 & \textnormal{\AND} \\[2pt]
					\subLP[z^{u',v'}_{\req,i,j}][(i,j),\restrict[\mapV][\VEbfsLabelsOrig]] > 0 & \textnormal{\FORALL } (u',v') \in {P}^{u,v}_{\req,i,j}
					\end{array}$\\
					  					\vspace{2pt}
					  \popline\popline
					  
						\SET $\mapE(i,j) \gets {P}^{u,v}_{\req,i,j}$ \AND \textbf{if} $\mapV(j) = \emptyset$ \textbf{then} $\mapV(j) \gets v$\\
						
					}{
						\COMPUTE path ${P}^{v,u}_{\req,j,i}$ from $v \in \VVloc[j]$ to $\mapV(i)=u$ according to Lemma~\ref{lem:local-connectivity-property}\\
						\pushline \pushline \nonl s.t.
						$\begin{array}{rl}
							\subLP[y^v_{\req,j}][\origEdge[i,j],\restrict[\mapV][\VEbfsLabels]] > 0 & \textnormal{\AND} \\[2pt]
							\subLP[z^{u',v'}_{\req,j,i}][\origEdge[i,j],\restrict[\mapV][\VEbfsLabels]] > 0 & \textnormal{\FORALL } (u',v') \in {P}^{u,v}_{\req,j,i}
						\end{array}$\\
						\vspace{2pt}
						\popline\popline

						\SET $\mapE(\origEdge[i,j]) \gets {P}^{u,v}_{\req,j,i}$ \AND \textbf{if} $\mapV(j) = \emptyset$ \textbf{then} $\mapV(j) \gets v$\\


					}
					\If{$\mapE(\origEdge[e]) \neq \emptyset$ \textnormal{\FORALL}$e \in \deltaMinusA[j]$}{
						\SET $\mathcal{Q} \gets \mathcal{Q} \cup \{j\}$\\
					}
				}
			}
		}
		\SET $\mathcal{V}_k \gets 
		\left( \begin{array}{ll}
						& \{x_\req\}  \cup  \{y^u_{\req,i} ~|~ i \in \VV, u=\mapV(i)\} \\
				\cup 	& \{~~\,\subLP[x_{\req}][e,\restrict[\mapV][\VEbfsLabelsOrig]] ~|~ e \in \VE \} \\
				\cup 	& \{~\,\subLP[y^u_{\req,i}][e,\restrict[\mapV][\VEbfsLabelsOrig]] ~|~ e \in \VE, i \in \VVe, u=\mapV(i) \}\\
				\cup  	& \{\subLP[z^{u,v}_{\req,i,j}][e,\restrict[\mapV][\VEbfsLabelsOrig]]~|~ e=(i,j) \in \VE, (u,v) \in \mapE(i,j)\} \\
				\cup 	& \{\gamma^{u}_{\req,i,b,a}~|~ i \in \VV, u = \mapV(i), \VEbfsBag \in \VEbfsBags, \mapVbag=\restrict[\mapV][\VEbfsLabelsBag]\} 
				\end{array}  \right)$ \label{alg:decomposition:compute-Vk}\\
		\SET $\prob \gets \min \Variables$ \label{alg:decomposition:computing-prob} \\
		\SET $v \gets v - \prob$ \FORALL $v \in \Variables$ \label{alg:decomposition:adapt-variables-one}\\
		\SET $a^{x,y}_{\req} \gets a^{x,y}_{\req} - \prob \cdot A(\mapping,x,y)$ \FORALL $(x,y) \in \SR$\label{alg:decomposition:adapt-load-one}\\	
		\ForEach{$(i,j) \in \VE$ \textnormal{\AND \EACH} $(x,y) \in \{(\Vtype(i),i), (\Vtype(j),j), (i,j)\}$}{
			\SET $\subLP[a^{x,y}_{\req}][(i,j),\restrict[\mapV][\VEbfsLabelsOrig]] \gets \subLP[a^{x,y}_{\req}][\restrict[\mapV][\VEbfsLabelsOrig]] - \prob \cdot A(\mapping,x,y)$ \label{alg:decomposition:adapt-load}\\
		}
		\ADD $\decomp = (f^k_{\req},m^k_{\req})$ to $\PotEmbeddings$ \AND \SET $k \gets k + 1$\\
	}

\KwRet{$\PotEmbeddings$}
\caption{Decomposition algorithm for solutions to the novel LP formulation~\ref{IP:novel-AC}}
\label{alg:decompositionAlgorithm-Novel-AC}
\end{algorithm*}
\end{minipage}}
\end{figure}

\subsection{Decomposition Algorithm for the Novel LP Formulation}

We now formally present the decomposition algorithm (see Algorithm~\ref{alg:decompositionAlgorithm-Novel-AC}) and prove its correctness.  The algorithm builds on the ideas of the decomposition algorithm for the MCF Formulation~\ref{alg:VNEP-IP-old} presented in Section~\ref{sec:decomposing-mcf-solution-trees}.

Fixing the mapping of the root initially, mappings for the outgoing edges (with respect to the extraction order) are extracted again together with the mappings of the heads of these edges using Lemma~\ref{lem:local-connectivity-property}. However, as the edge embeddings are computed using a copy of the MCF formulation for each node mapping function of the edge's labels, the mapping of the edge's labels to substrate nodes must be fixed first. To this end, we employ the node mapping variables $\gamma^u_{\req,i,b,a}$ of the edge bags. Concretely, whenever the outgoing edges of $i \in \VV$ are mapped, we show that given the mapping of the virtual node $i$ onto some substrate node $\mapV(i)=u$, we can always find a variable $\gamma^u_{\req,i,b,a}$ for edge bag $\VEbfsBag \in \VEbfsBags$ and $\mapVbag \in \MappingSpace[\VEbfsLabelsBag]$, such that (i) the mapping of the bag $\mapVbag$ agrees with the previous node mappings, i.e., $\mapVbag(i') = \mapV(i')$ holds for all previously mapped virtual nodes $i'$, and that (ii) the variable $\gamma^u_{\req,i,b,a}$ is strictly greater 0. Given such a variable $\gamma^u_{\req,i,b,a}$ and fixing the node mappings of all the new labels contained in $\VEbfsLabels$ according to $\mapVbag$, mappings for the outgoing edges can always be extracted due to Constraint~\ref{LP:novel:gamma-to-outgoing-edges}. Concretely, after having fixed the mappings of the labels of the outgoing edge $e$, Constraint~\ref{LP:novel:gamma-to-outgoing-edges} ensures that for the sub LP with index $\restrict[\mapV][\VEbfsLabels]$ the condition $y^u_{\req,i} > 0$ holds, thereby allowing the application of Lemma~\ref{lem:local-connectivity-property} to extract the mapping for the edge $e$. Having given this intuition, we now formally prove its correctness.

\begin{restatable}{theorem}{decomposabilityNovelFormulation}
\label{thm:decomposability-novel-formulation}
For a request $\req \in \requests$ and its extraction order $\VGbfs$, a solution to Formulation~\ref{IP:novel-AC} can be decomposed into $\PotEmbeddings =\{(\prob,\mapping)\}_k$ in time $\mathcal{O}(|\SG|^{2\cdot \ewX(\VGbfs)+1}  \cdot |\VG|^2)$, such that:
\begin{itemize}
\item The decomposition is complete, i.e., $x_{\req} = \sum_k \prob$ holds.
\item Allocations are bounded by $\vec{a}_{\req}$, i.e., $a^{x,y}_{\req} \geq \sum_{(\prob,\mapping)} \prob \cdot A(\mapping,x,y)$ holds for $(x,y) \in \SR$.
\end{itemize}
\end{restatable}
\begin{proof}
We prove that each iteration yields a valid mapping $\mapping$ of value $\prob > 0$.

First, note that if in Line~10 a suitable mapping $\mapVbag$ was found, such that the respective edge bag variable $\gamma^{\mapV(i)}_{\req,i,b,a}$ is positive, then the requirement of Lemma~\ref{lem:local-connectivity-property} that $\subLP[y^{\mapV(i)}_{\req,i}][e,\restrict[\mapV][\VEbfsLabels]]>0$ holds is always satisfied due to Constraint~\ref{LP:novel:gamma-to-outgoing-edges}.

The initial mapping of the root in Line~5 is always possible due to Constraints~\ref{LP:novel:node-embedding}. Furthermore, when considering the edge bags of the root $\VVroot$, there will always exist a suitable edge bag variable $\gamma^{\mapV(\VVroot)}_{\req,\VVroot,b,a} > 0$ to choose from due to Constraints~\ref{LP:novel:node-to-sub-node-mapping} and \ref{LP:novel:gamma-to-outgoing-edges}. 

Having chosen a suitable mapping for the labels of the edge bag, the extraction of mappings for the outgoing edges $e \in \VEbfsBag$ is always feasible, as Constraint~\ref{LP:novel:gamma-to-outgoing-edges} induces $y^{\mapV(\VVroot)}_{\req,\VVroot} > 0$ in the respective sub LPs. Also note that the application of Lemma~\ref{lem:local-connectivity-property} safeguards that the head $j$ is mapped positively on some substrate node $i$, i.e. we have $\subLP[y^v_{\req,j}][\restrict[\mapV][\VEbfsLabels]]\,>\,0$.

Given the initial validity of the mapping of the root and its outgoing edges, assume now for the sake of contradiction that the extraction process fails at some point in time. Concretely, we consider the first point in time at which the constructed (partial) mapping $\mapping=(\mapV,\mapE)$ is not valid anymore or at which the choose operation in Line~10 fails. 

We first consider the case that the mapping $\mapping$ is not valid (anymore), such that the mapping of an edge $e=(i,j)$ fails to start at $\mapV(i)$ or fails to lead to $\mapV(j)$. Edges are only mapped in Lines~15 and 18 and we consider w.l.o.g. that the algorithm fails in Line~15. For this type of failure to happen, the node $j$ must have been mapped before as the node $j$ is otherwise validly mapped by the same line of the pseudocode. As $j$ can only be mapped multiple times if $j$ is itself a label and all incoming edges of a node share the same labels (see Lemma~\ref{lem:labels:in-edge-labeled-same}), the edge $(i,j)$ must have been labeled by $j$. Let $v'$ denote the substrate node in which the path $\mapE(i,j)$ ends and for which $v'\neq \mapV(i)$ holds. As stated in the beginning of the proof, the requirement of Lemma~\ref{lem:local-connectivity-property} is always valid (if a suitable mapping $\mapVbag$ was found) and hence, by applying Lemma~\ref{lem:local-connectivity-property} we obtain $\subLP[y^{v'}_{\req,j}][e,\restrict[\mapV][\VEbfsLabelsOrig]] > 0$. However, Constraint~\ref{LP:novel:forbidding-nodes-in-sub-lps} clearly forbids the use of this node $v'$ as it does not equal $\mapV(i)$ by setting $\subLP[y^{v'}_{\req,j}][e,\restrict[\mapV][\VEbfsLabelsOrig]] = 0$. This is a contradiction, and the only option for the extraction process to fail is hence due to an infeasible choose operation in Line 10.

As argued in the beginning of the proof, the choose operation may not fail for the root. Hence, the node $i \in \VV$ for which Line~10 fails, is not the root and has been reached by at least one incoming edge $(k,i) \in \VEbfs$. Assume that the choose operation fails for a specific edge bag $\VEbfsBag \in \VEbfsBags$.

We first show that $\VEbfsLabelsBag \cap M^V \subseteq \VEbfsLabels[(k,i)]$ holds. Assume for the sake of contradiction, that there exists some label $l \in \VEbfsLabelsBag \cap M^V$ such that $l \notin \VEbfsLabels[(k,i)]$. As the label $l$ is contained in $M^V$, a mapping was decided for $l$ at some other node $s_l \in \VV$. In particular, Lemma~\ref{lem:labels:unique-root-for-label} specifies that the node $s_l$ is the unique root of all the confluences towards $l$, such that any other node with an edge being labeled by $l$ must be reachable from $s_l$. However, as all incoming labels agree on their labels (see Lemma~\ref{lem:labels:in-edge-labeled-same}) and no incoming edge of the node $i$ hence lies on a confluence with target $l$, we must have $i=s_l$. In this case however, the node mapping of $l$ cannot have been decided before as the algorithm only fixes these node mappings once the choose operation was executed at the respective node $s_l$. Hence, $\VEbfsLabelsBag \cap M^V \subseteq \VEbfsLabels[(k,i)]$ holds.

As all previous mapping steps have been valid and the mappings were obtained by the application of Lemma~\ref{lem:local-connectivity-property}, we know that $\subLP[y^{u}_{\req,i}][\origEdge[(k,i)],\restrict[\mapV][\VEbfsLabels[(k,i)]]] > 0$ holds for $u=\mapV(i)$. Consider now the particular mapping $\mapVedge[(k,i)] = \restrict[\mapV][\VEbfsLabelsEdgeCapBag[b][(k,i)]]$ which is well-defined, as all labels of the incoming edge $(k,i)$ must have been fixed before extracting the mapping of this edge. From the validity of Constraint~\ref{LP:novel:incoming-edges-to-gamma-variables} we obtain that
$\sum_{\mapVbag \in \MappingSpace[\VEbfsLabelsBag]: \restrict[\mapVbag][\VEbfsLabelsEdgeCapBag] =  \mapVedge[(k,i)]} \gamma^u_{\req,i,b,a} > 0$ holds, as $\subLP[y^{u}_{\req,i}][\origEdge[(k,i)],\restrict[\mapV][\VEbfsLabels[(k,i)]]]$ is larger than 0.
As $\VEbfsLabelsBag \cap M^V \subseteq \VEbfsLabels[(k,i)]$ holds, we know that $\sum_{\mapVbag \in \MappingSpace[\VEbfsLabelsBag]: \restrict[\mapVbag][\VEbfsLabelsBag \cap M^V] = \restrict[\mapV][\VEbfsLabelsBag \cap M^V]} \gamma^u_{\req,i,b,a} > 0$ holds, since the restriction in the sum's index has been loosened. Hence, the choose operation in Line~10 can always be successfully executed and the mapping constructed in the $k$-th iteration will always be valid.
 
It is easy to check that the claims with respect to the completeness and the bounds by the load variables also hold: the mapping is always covered by respective mapping variables in $\mathcal{V}_k$ and as the load is computed as a function of these mapping variables, the extracted fractional resource allocations are also bounded by $\vec{a}_{\req}$.

\newcommand{\CHOOSE}{\textbf{choose~}}
Lastly, every time a valid mapping is extracted, a mapping variable's value is set to $0$. As the formulation has size $\mathcal{O}(|\SG|^{\ewX(\VGbfs)} \cdot |\VG|)$ (cf. Theorem~\ref{thm:size-of-formulation-is-exponential-in-ex-width}) for request $\req$, and this also bounds the number of variables, at most $\mathcal{O}(|\SG|^{\ewX(\VGbfs)} \cdot |\VG|)$ many valid mappings may be recovered. For recovering a single valid mapping, the runtime can be  bounded by $\mathcal{O}(|\SG|^{\ewX(\VGbfs)+1} \cdot |\VG|)$, as the \CHOOSE operation in Line~10 is executed at most $|\VV|$ times and the path computations in Lines~14 and 17 can be implemented in time $\mathcal{O}(|\SE|)$ (cf.~Lemma~\ref{lem:local-connectivity-property}). Hence, the overall runtime of the decomposition algorithm is bounded by $\mathcal{O}(|\SG|^{2\cdot \ewX(\VGbfs)+1} \cdot |\VG|^2)$.

\end{proof}

\section{FPT-Approximations for the Virtual Network Embedding Problem}
\label{sec:approximation-via-randround}

As shown above, our novel LP Formulation~\ref{IP:novel-AC} is sufficiently strong, such that its solutions can be decomposed into convex combinations $\PotEmbeddings = \{(\prob,\map)\}_k$ for each request $\req \in \requests$. In this section we now apply randomized rounding~\cite{Raghavan-Thompson} on the decomposed solution to obtain fixed-parameter tractable tri-criteria approximations for the profit and the cost variant of the VNEP.

In the following, we cast the quality of the found solutions in terms of random variables to bound the respective probabilities of not finding a suitable solution. To this end, we employ the following well-known tail bounds.

\begin{theorem}[Chernoff-Bound \cite{dubhashi2009concentration}]
\label{thm:chernoff}
Let $X = \sum_{i = 1}^n X_i$ be a sum of $n$ independent random variables $X_i \in [0,1]$. Then $
\mathbb{P} \big( X \leq (1-\varepsilon)\cdot \mathbb{E}(X)~\big)~\leq exp(-\varepsilon^2\cdot \mathbb{E}(X)/2)$ holds for $0 < \varepsilon < 1$.
\end{theorem}
\begin{theorem}[Hoeffding's Inequality \cite{dubhashi2009concentration}]
\label{thm:hoeffding}
Given independent random variables $\{X_i\}_{i}$, s.t.~$X_i \in [a_i,b_i]$, then $\mathbb{P} \Big(\sum_i X_i - \mathbb{E}(\sum_i X_i))\geq t\Big) \leq exp \big(-2t^2 /~(\sum_i~(b_i - a_i)^2)\big)$ holds.
\end{theorem}

\subsection{Approximating the Profit Variant}

The pseudo-code of our approximation for the profit is presented as Algorithm~\ref{alg:rand-round-profit}. The algorithm first performs preprocessing in Lines 1-3 by removing all requests which cannot be fully (fractionally) embedded in the absence of other requests. As these requests cannot be fully embedded, these requests can never be part of any feasible solution and can hence be removed. In Lines 4-8 the randomized rounding scheme is applied: an LP solution to the Formulation~\ref{IP:novel-AC} is computed, decomposed and then rounded. The rounding procedure is iterated as long as the constructed solution is not of sufficient quality or until the number of maximal rounding tries is exceeded. Concretely, we seek $(\alpha,\beta,\gamma)$-approximate solutions which achieve at least a factor of $\alpha<1$ times the optimal (LP) profit and exceed node and edge capacities by at most factors of $\beta\geq1$ and $\gamma \geq1$, respectively. In the following we discuss the parameters $\alpha$, $\beta$, and $\gamma$ for which solutions can be found \emph{with high probability}.

 \begin{figure}[t!]
\scalebox{0.93}{
 \begin{minipage}{1.07\columnwidth}
 
 \removelatexerror
 
 \begin{algorithm*}[H]

 \SetKwInOut{Input}{Input}\SetKwInOut{Output}{Output}
 \SetKwFunction{ProcessPath}{ProcessPath}{}{}
 \SetKwFunction{reverse}{reverse}{}{}
 \SetKwFunction{LP}{LP}
 \SetKwFunction{LP}{LP}
 
 \newcommand{\SET}{\textbf{set~}}
 \newcommand{\ADD}{\textbf{add~}}
 \newcommand{\DEFINE}{\textbf{define~}}
 \newcommand{\AND}{\textbf{and~}}
 \newcommand{\LET}{\textbf{let~}}
 \newcommand{\WITH}{\textbf{with~}}
 \newcommand{\COMPUTE}{\textbf{compute~}}
  \newcommand{\CHECK}{\textbf{check~}}
    \newcommand{\REMOVE}{\textbf{remove~}}
     \newcommand{\REPEAT}{\textbf{repeat~}}
 \newcommand{\FIND}{\textbf{find~}}
 \newcommand{\CHOOSE}{\textbf{choose~}}
  \newcommand{\CHOOSi}{\textbf{choosing~}}
  \newcommand{\CONS}{{construct solution by~}}
 \newcommand{\DECOMPOSE}{\textbf{decompose~}}
 \newcommand{\FORALL}{\textbf{for all~}}
 \newcommand{\OBTAIN}{\textbf{obtain~}}
 \newcommand{\WITHPROBABILITY}{\textbf{with probability~}}
 
 \SetKwRepeat{Do}{do}{while}%
 
\ForEach(\tcp*[f]{preprocess requests}){$\req \in \requests$}{ 
	\COMPUTE LP solution to Formulation~\ref{IP:novel-AC} for the request set $\{\req\}$ maximizing $x_r$\\
	\REMOVE request $\req$ from the set $\requests$ if $x_r < 1$ holds
}
\COMPUTE LP solution to Formulation~\ref{IP:novel-AC} for request set $\requests$ maximizing $\sum_{\req \in \requests} \Vprofit \cdot x_{\req}$\\
\DECOMPOSE LP solution into convex combinations $\PotEmbeddings$ \FORALL $\req \in \requests$\\
\Do(\tcp*[f]{perform randomized rounding}){\textnormal{the solution is \emph{not} $(\alpha,\beta,\gamma)$-approximate and maximal rounding tries are not exceeded}}{
	\CONS \CHOOSi mapping $\mapping$ \WITHPROBABILITY $\prob$ \FORALL $\req \in \requests$\\
}
 \caption{Randomized Rounding Algorithm for the VNEP (Profit)}
 \label{alg:rand-round-profit}
 \end{algorithm*}
 \end{minipage}}
 
 \end{figure}
 
\noindent\textbf{Bounding the Profit.}
Employing the \emph{discrete} random variable $\randVarY \in \{0,\Vprofit\}$ to model the profit achieved by (potentially)~embedding request $\req \in \requests$, we have  $\ProbVarY[\Vprofit] = \sum_{(\prob,\mapping) \in \PotEmbeddings} \prob$ and $\ProbVarY[0] = 1-\sum_{(\prob,\mapping) \in \PotEmbeddings} \prob$. Hence, the overall profit achieved is  $B = \sum_{\req \in \requests} \randVarY$ with $\mathbb{E}(B) = \sum_{\req \in \requests} \Vprofit \cdot \sum_{(\prob,\mapping) \in \PotEmbeddings} \prob $. As the decomposition is complete (cf.~Theorem~\ref{thm:decomposability-novel-formulation}) we have $\optLP = \Exp(B)$, where $\optLP$ denotes the profit of the LP solution.

By removing any request, which cannot be fully fractionally embedded (in the absence of other requests) in the preprocessing step, we know that $\optLP \geq \max_{\req \in \requests} \Vprofit$ holds. By applying the Chernoff-Bound of Theorem~\ref{thm:chernoff}, we obtain:
\begin{restatable}{lemm}{probabilityOfNotSucceedingInObjective}
\label{thm:probability-of-not-succeeding-in-objective}
The probability of achieving less than $1/3$ of the profit of the optimal solution is upper bounded by $exp(-2/9) \approx 0.8007$.
\end{restatable}
\begin{proof}
Let $\hat{b} = \max_{\req \in \requests} \Vprofit$ denote the maximum benefit of the requests.
We consider the random variables $Y'_{\req} = Y_{\req} / \hat{b}$, such that $Y'_{\req} \in [0,1]$ holds.
Let $B' = \sum_{\req \in \requests} Y'_{\req}$ denote the total profit achieved after scaling down the profits. 
As $\mathbb{E}(B) = \optLP \geq \hat{b}$ holds, we have $\mathbb{E}(B') \geq 1$
Choosing $\varepsilon = 2/3$ and applying Theorem~\ref{thm:chernoff} on $B'$ we obtain $\mathbb{P} \big( B' \leq (1/3)\cdot \mathbb{E}(B')~\big)~\leq exp(-2\cdot \mathbb{E}(B')/9)$.
Plugging in the \emph{minimal} value of $\Exp(B')$, i.e., $1$, into the equation we obtain:
$\mathbb{P} \big( B' \leq (1/3)\cdot \mathbb{E}(B')~\big) \leq exp(-2/9)$ and accordingly 
$\mathbb{P} \big( B \leq (1/3)\cdot \mathbb{E}(B)~\big) \leq exp(-2/9)$.

As mentioned before, by Theorem~\ref{thm:decomposability-novel-formulation} we have $\Exp(B)= \optLP$. 
Denoting the optimal profit of the Integer Program  by $\optIP$ and observing that $\optIP \leq \optLP$ holds as the IP is contained in the solution space of the LP, we have
$\optIP/3 \leq \optLP/3 = \Exp(B)/3$. Hence, we obtain $\mathbb{P} \big( B \leq (1/3) \cdot \optIP \big) \leq \mathbb{P} \big( B \leq (1/3) \cdot \Exp(B)/3 \big) \leq exp(-2/9)$\,,
completing the proof.
\end{proof}

\noindent\textbf{Bounding Resource Allocations.} We model the allocations of request $\req \in \requests$ on resource $(x,y) \in  \SR$ as random variable $A_{\req,x,y} \in [0,\maxAllocX]$. We have $\mathbb{P}(A_{\req,x,y} = A(\mapping, x,y))= \prob $ and $\mathbb{P}(A_{\req,x,y} = 0)=  1 - \sum_{(\prob,\mapping) \in \PotEmbeddings } \prob $.
Let $A_{x,y} = \sum_{\req \in \requests} A_{\req,x,y}$ denote the cumulative allocations induced on resource $(x,y)$. As $\mathbb{E}(A_{x,y}) = \sum_{\req \in \requests} \sum_{(\prob,\mapping) \in \PotEmbeddings} \prob \cdot A(\mapping,x,y)$ holds by definition and using Theorem~\ref{thm:decomposability-novel-formulation},  $\mathbb{E}(A_{x,y}) \leq \Scap(x,y)$ is obtained for $(x,y) \in  \SR$.
\begin{lemma}
\allowdisplaybreaks
\label{lem:approximation-single-node}
Consider a node resource $(\tau,u) \in \SRV$. Choose $0< \varepsilon \leq 1$, such that $\maxDemandV / \Scap(\type,u) \leq \varepsilon$ holds for $\req \in \requests$. Let~$\Delta = \sum_{\req \in \requests} (\maxAllocV / \maxDemandV)^2$ and $\beta = 1+\varepsilon \cdot  \sqrt{2\cdot \Delta \cdot \log(|\SV|\cdot|\types|)}$. Then $\mathbb{P} \Big( \hspace{-2pt} A_{\type,u} \geq \beta \cdot \Scap(\type,u) \hspace{-2pt} \Big) \hspace{-2pt} \leq \hspace{-2pt}(|\SV|\cdot |\types|)^{-4}$ holds.
%
\begin{proof}
\allowdisplaybreaks

We apply Hoeffding (cf.~Theorem~\ref{thm:hoeffding}) with $t = \varepsilon \cdot \sqrt{2\cdot \Delta \cdot \log(|\SV|\cdot|\types|)  } \cdot \Scap(\type,u)$:
{
\small
\begin{alignat*}{5}
\hspace{-12pt}
\mathbb{P} \Big(A_{\type,u} - \mathbb{E}(A_{\type,u})~\geq t \Big) &  \leq && exp \Big(\frac{-4 \cdot \varepsilon^2 \cdot  \Delta \cdot \log(|\SV|\cdot|\types|) \cdot  \Scap^2(\type,u)}{\sum_{\req \in \requests}~(\maxAllocV)^2} \Big) \\
& \leq && exp \Big(\frac{-4 \cdot \varepsilon^2  \cdot \Delta  \cdot \log (|\SV|\cdot|\types|) \cdot \Scap^2(\type,u) }{\sum_{\req \in \requests}~(\varepsilon \cdot \Scap(\type,u) \cdot \maxAllocV / \maxDemandV)^2} \Big)\\
& ~ = &&exp \Bigg(\frac{-4 \log~(|\SV|\cdot|\types|)~\sum_{\req \in \requests}~(\maxAllocV / \maxDemandV)^2 }{\sum_{\req \in \requests}~( \maxAllocV / \maxDemandV)^2} \Bigg) \\
& = && (|\SV|\cdot |\types|)^{-4}
\end{alignat*}
}
Above, we have used $\maxAllocV \leq \varepsilon \cdot \Scap(\type,u)~\cdot \maxAllocV / \maxDemandV$, which follows from  $\maxDemandV \leq \varepsilon \cdot \Scap(\type,u)$. We then plugged in the definition of $\Delta$ and reduced the fraction. Lastly, to obtain the lemma's statement, we utilize that the expected allocations $\mathbb{E}(A_{\type,u})$ are upper bounded by the capacity $\Scap(\tau,u)$:
$\mathbb{P} \Big(A_{\type,u} \geq \beta \cdot \Scap(\type,u)~ \Big) \leq ~(|\SV|\cdot |\types|)^{-4}~.$
\end{proof}
\end{lemma}

The probability to violate edge resources can be bounded analogously (see Appendix~\ref{app:deferred-proofs} for the proof):
\begin{restatable}{lemm}{approximationSingleEdge}
\label{lem:approximation-single-edge}
We consider a single edge $(u,v) \in \SE$. Choose $0< \varepsilon \leq 1$, such that $\maxDemandE / \Scap(u,v) \leq \varepsilon $ holds for all $\req \in \requests$. 
Let $\Delta = \sum_{\req \in \requests} (\maxAllocE / \maxDemandE)^2$ and $\gamma = 1+\varepsilon\cdot \sqrt{2\cdot \Delta \cdot \log~|\SV|}$. Then $\mathbb{P} \Big( A_{u,v} \geq \gamma \cdot \Scap(u,v) \Big)  \leq |\SV|^{-4}$ holds.
\end{restatable}

\noindent\textbf{Main Result.}
Given the above, we can now prove that Algorithm~\ref{alg:rand-round-profit} indeed is a FPT approximation for the profit variant of the VNEP \emph{with high probability}.

\begin{restatable}{theorem}{mainResultForAdmissionControl}
\label{thm:result-for-admission-control}
Assume $|\SV| \geq 3$. Let $0 < \varepsilon \leq 1$ be chosen minimally, such that  $\maxDemandV[\req][x][y] / \Scap(x,y) \leq \varepsilon$ holds for $(x,y) \in  \SR$ and $\req \in \requests$. Randomized rounding yields a~$(\alpha,\beta,\gamma)$ tri-criteria FPT approximation for the profit variant of the VNEP,  such that it finds a solution \emph{with high probabiliy} of at least an~$\alpha = 1/3$ fraction of the optimal profit, and cumulative allocations within factors of $\beta$ (nodes) and $\gamma$ (edges) of the original capacities, for $\beta = 1+\varepsilon \cdot \sqrt{2 \cdot \Delta(\SRV) \cdot \log(|\SV|\cdot|\types|)}$ and $\gamma  = 1 +  \varepsilon \cdot \sqrt{2 \cdot \Delta(\SE) \cdot \log |\SV| }$ with $\Delta: \mathcal{P}(\SR) \to \mathbb{R}_{\geq 0}$ being defined as $\Delta(X) = \max_{(x,y) \in X} \sum_{\req \in \requests} \left(\frac{\maxAllocX}{ \maxDemandX} \right)^2$.
\end{restatable}  
\begin{proof}[Proof.]
We consider the probability of the rounding step failing to produce a $(\alpha,\beta,\gamma)$-approximate solution. By applying a union bound for any node and edge resource exceeding $\beta$ or $\gamma$ times the capacity (cf.~Lemma~\ref{lem:approximation-single-node} and Lemma~\ref{lem:approximation-single-edge}) together with the probability of \emph{not} achieving $1/3$ of the LP's objective (cf. Lemma~\ref{thm:probability-of-not-succeeding-in-objective}), the probability to not find a solution is upper bounded by $19/20$. Hence, the probability of finding an approximate solution within $N$ iterations is at least $1-(19/20)^N$, i.e., Algorithm~\ref{alg:rand-round-profit} will produce such a solution \emph{with high probability}.
With respect to the runtime of Algorithm~\ref{alg:rand-round-profit} we note that the LP Formulation~\ref{IP:novel-AC} can be solved in $\mathcal{O}\left(poly(\sum_{\req \in \requests} |\VG| \cdot |\SG|^{\ewX(\VGbfs)})\right)$ by e.g. using the Ellipsoid algorithm~\cite{grotschel1988geometric} (cf.~\ref{thm:size-of-formulation-is-exponential-in-ex-width}) and the decomposition algorithm's runtime has the same runtime. Hence, the runtime of Algorithm~\ref{alg:rand-round-profit} is bounded by $\mathcal{O}\left(N \cdot poly(\sum_{\req \in \requests} |\VG| \cdot |\SG|^{\ewX(\VGbfs)})\right)$, when performing at most $N$ rounding tries.
Hence Algorithm~\ref{alg:rand-round-profit} is fixed-parameter tractable with respect to the maximal width of any of the extraction orders.
\end{proof}
\subsection{Approximating the Cost Variant} 
 {
  \renewcommand{\optLP}{\ensuremath{\textnormal{C}_{\textnormal{LP}}}}
   \renewcommand{\optIP}{\ensuremath{\textnormal{C}_{\textnormal{opt}}}}
   
     \begin{figure}[tb!]
     \scalebox{0.93}{
     \begin{minipage}{1.07\columnwidth}
     
     \removelatexerror
     
     \begin{algorithm*}[H]

     \SetKwInOut{Input}{Input}\SetKwInOut{Output}{Output}
     \SetKwFunction{ProcessPath}{ProcessPath}{}{}
     \SetKwFunction{reverse}{reverse}{}{}
     \SetKwFunction{LP}{LP}
     \SetKwFunction{LP}{LP}
     
     \newcommand{\SET}{\textbf{set~}}
     \newcommand{\ADD}{\textbf{add~}}
     \newcommand{\DEFINE}{\textbf{define~}}
     \newcommand{\AND}{\textbf{and~}}
     \newcommand{\LET}{\textbf{let~}}
     \newcommand{\WITH}{\textbf{with~}}
     \newcommand{\COMPUTE}{\textbf{compute~}}
      \newcommand{\CHECK}{\textbf{check~}}
        \newcommand{\REMOVE}{\textbf{remove~}}
         \newcommand{\REPEAT}{\textbf{repeat~}}
     \newcommand{\FIND}{\textbf{find~}}
     \newcommand{\CHOOSE}{\textbf{choose~}}
      \newcommand{\CHOOSi}{\textbf{choosing~}}
      \newcommand{\CONS}{{construct solution by~}}
     \newcommand{\DECOMPOSE}{\textbf{decompose~}}
     \newcommand{\NORMALIZE}{\textbf{normalize~}}
     \newcommand{\FORALL}{\textbf{for all~}}
     \newcommand{\OBTAIN}{\textbf{obtain~}}
     \newcommand{\WITHPROBABILITY}{\textbf{with probability~}}
     
     \SetKwRepeat{Do}{do}{while}%

    \COMPUTE LP solution to Formulation~\ref{IP:novel-AC} for request set $\requests$ subject to \\
    \pushline \pushline \nonl 
    						\begin{tabular}{lcl}
   						 ~~~the objective & $\min \sum_{(x,y) \in \SR, \req \in \requests} \Scost(x,y) \cdot a^{x,y}_{\req}$ & and \\
   						 ~~~the additional constraints & $x_{\req} = 1$ & for all $\req \in \requests$
    						\end{tabular}\\
    \popline \popline \nl
    \DECOMPOSE LP solution into convex combinations $\PotEmbeddings$ \FORALL $\req \in \requests$\\
    \ForEach(\tcp*[f]{postprocess decomposition: prune costly mappings}){$\req \in \requests$}{ 
       \LET $\WAC = \sum_{(\prob, \mapping) \in \PotEmbeddings} \prob \cdot c(\mapping)$\\
     	\REMOVE tuples $(\prob,\mapping)$ from $\PotEmbeddings$ with $c(\mapping) > 2\cdot \WAC$ \\
     	\NORMALIZE weights of $\PotEmbeddings$, such that $\sum_{(\prob, \mapping)} \prob = 1$ holds again\\
     }
   
   \Do(\tcp*[f]{perform randomized rounding}){\textnormal{the solution is \emph{not} $(\alpha,\beta,\gamma)$-approximate and maximal rounding tries are not exceeded}}{
    	\CONS \CHOOSi mapping $\mapping$ \WITHPROBABILITY $\prob$ \FORALL $\req \in \requests$\\
    }
     \caption{Randomized Rounding Algorithm for the VNEP (Cost)}
     \label{alg:rand-round-cost}
     \end{algorithm*}
    \end{minipage}}
     
     \end{figure}
     
Similarly to the approximation of the profit, the cost variant can be approximated as presented in Algorithm~\ref{alg:rand-round-cost}. In particular, two changes are necessary compared to Algorithm~\ref{alg:rand-round-profit}: the preprocessing in Lines 1-3 can be dropped and a postprocessing of the found decompositions needs to be added. 
  Concretely, to obtain a constant approximation of the cost, the decision which of the mappings $\mapping$ to choose for request $\req \in \requests$ cannot be purely left to chance. Denoting the cost of a mapping $\mapping \in \spaceSolReq$ by $c(\mapping) = \sum_{(x,y) \in \SR} \Scost(x,y) \cdot A(\mapping,x,y)$, the set of convex combinations $\PotEmbeddings = \{(\prob, \mapping)\}_k$ may contain mappings $\mapping$ of arbitrarily small weight $\prob$ while having an arbitrarily high cost $c(\mapping)$. Hence, if the possibility exists to choose such a mapping in the rounding step, no bound on the rounded cost can be given in general. Hence, the key idea is to remove all fractional mappings of high cost while not losing too much weight in the convex combination.
 Concretely, given a request $\req \in \requests$, we denote by $\WAC = \sum_{\decomp \in \PotEmbeddings} \prob \cdot c(\mapping)$ the \emph{weighted (averaged)~cost} of request $\req \in \requests$ and the algorithm removes all mappings $\mapping$ from the convex combinations for which $c(\mapping) > 2 \cdot \WAC$ hold.
  As the following lemma shows, the sum of the associated weights of the mappings removed must be less than $1/2$:
 
 \begin{lemma}
 \label{lem:wac-pruning}
 The sum of the weights $\prob$ of the mappings $\mapping$ with cost smaller than two times $\WAC$ is at least $1/2$ for each request $\req \in \requests$.
  \end{lemma}
 \begin{proof}
 Let $\lambda_{\req} = \sum_{(\prob,\mapping) \in \PotEmbeddings : c(\mapping)~\leq 2 \cdot \WAC } \prob$ denote the 
 sum of the weights of the mappings of cost bounded by $2\cdot \WAC$.
 For the sake of contradiction, assume that $\lambda_{\req} < 1/2$ holds for any request $\req \in \requests$. By the definition of $\WAC$ and the assumption on $\lambda_{\req}$, we obtain the following contradiction:
 \begin{alignat}{5}
 & \WAC ~& = & ~\sum_{(\prob,\mapping)  \in \PotEmbeddings} \prob \cdot c(\mapping)~\label{eq:a}\\
 &      & \leq & ~\sum_{(\prob,\mapping)  \in \PotEmbeddings: c(\mapping)> 2 \cdot \WAC} \prob \cdot c(\mapping)~\label{eq:b}\\
 &      & \leq  & ~\sum_{(\prob,\mapping)  \in \PotEmbeddings: c(\mapping)> 2 \cdot \WAC} \prob \cdot 2\cdot \WAC \label{eq:c}\\
 &      & \leq  & ~(1-\lambda_{\req})~\cdot 2 \cdot \WAC <  ~\WAC \label{eq:d}
 \end{alignat}
The validity of Equation~\ref{eq:c} follows as all the considered decompositions have a cost of at least two times $\WAC$ and the first inequality of Equation~\ref{eq:d} follows as $(1-\lambda_{\req})~> 1/2$ holds by assumption. Lastly, Equation~\ref{eq:d} yields the contradiction, showing that indeed $\lambda_{\req} \geq 1/2$ holds for $\req \in \requests$.
 \end{proof}

\noindent\textbf{Deterministic Guarantee for the Cost.}
It is easy to establish that initially, i.e., before pruning mappings from $\PotEmbeddings$, the cost of an optimal solution equals the sum of the (fractional) costs of the convex combinations:

\begin{lemma}
\label{cor:relation-of-net-profit-in-decomposition-and-the-LP-cost-new}
Letting $\optLP$ denote the cost of the LP solution computed in Line~1, we have:
\begin{align}
\sum_{\req \in \requests} \WAC &= \sum_{\req \in \requests, \decomp \in \PotEmbeddings} \prob \cdot c(\mapping) = \optLP~.
\end{align}
\end{lemma}
\begin{proof}
Due to the definition of the objective, we have $\optLP = \sum_{(x,y) \in \SR, \req \in \requests} \Scost(x,y) \cdot a^{x,y}_{\req}$. From the correctness of the decomposition algorithm (cf.~Theorem~\ref{thm:decomposability-novel-formulation}) we know that $a^{x,y}_{\req} \geq \sum_k \prob \cdot A(\mapping,x,y)$ holds for each request $\req \in \requests$ and each resource $(x,y) \in \SR$. As the cost $c(\mapping)$ of a mapping $\mapping$ equals $\sum_{(x,y) \in \SR} \Scost(x,y) \cdot A(\mapping,x,y)$, the equality  $\sum_{\req \in \requests, \decomp \in \PotEmbeddings} \prob \cdot c(\mapping) = \optLP$ follows. Lastly, as the equality $\sum_{\req \in \requests} \WAC = \sum_{\req \in \requests, (\prob, \mapping)  \in \PotEmbeddings} \prob \cdot c(\mapping)$ holds by definition of $\WAC$.
\end{proof}

\begin{lemma}
\label{lem:deterministic-objective-guarantee-without}
The cost of any solution returned by the randomized rounding scheme is upper bounded by two times the optimal cost.
\begin{proof}
  \vspace{-6pt}
Let $\optLP$ denote the cost of the optimal solution returned by LP Formulation~\ref{IP:novel-AC}. By allowing to select only decompositions $\mapping$ for which $c(\mapping)~\leq 2\cdot \WAC$ holds and denoting the selected mapping by $\hat{m}$ we have $c(\hat{m}_{\req})~\leq 2\cdot \WAC$. Hence $\sum_{\req \in \requests} c(\hat{m}_{\req})~\leq 2\cdot \sum_{\req \in \requests} \WAC$ holds. Together with Lemma~\ref{cor:relation-of-net-profit-in-decomposition-and-the-LP-cost-new} we obtain $\sum_{\req \in \requests} c(\hat{m}_{\req}) \leq 2\cdot \optLP$. As $\optLP$ is a lower bound for the minimum cost $\optIP$ of the optimal solution the lemma holds.
\end{proof}
\end{lemma}

\noindent\textbf{Bounding Resource Allocations.}
The employed rounding scheme for the cost variant of the VNEP equals the one for the profit. Revisiting the analysis of the probabilistic guarantees on the capacity violations for the profit approximation, we note that the analysis only depended on the fact that the expected load on any resource is upper bounded by its capacity. Due to the rescaling, this does not hold anymore. However, as the weights are scaled by at most a factor of $2$ (cf.~Lemma~\ref{lem:wac-pruning}), it is easy to establish the following:
\begin{lemma}
\label{lem:expected-load-less-than-cap-lala}
$\mathbb{E}(A_{x,y}) \leq 2\cdot \Scap(x,y)$ holds for all resources $(x,y) \in  \SR$.
\end{lemma}
Plugging in this doubled expected load into the proofs of Lemmas~\ref{lem:approximation-single-node} and \ref{lem:approximation-single-edge}, the approximation factors become $\beta = 2+\varepsilon\cdot \sqrt{2\cdot \Delta \cdot \log(|\SV|\cdot|\types|)}$ and $\gamma = 2+\varepsilon\cdot \sqrt{2\cdot \Delta \cdot \log(|\SV|)}$, respectively, and we state the following lemmas without proof.

\begin{lemma}
\allowdisplaybreaks
\label{lem:approximation-single-node-cost}
Consider a node resource $(\tau,u) \in \SRV$. Choose $0< \varepsilon \leq 1$, such that $\maxDemandV / \Scap(\type,u) \leq \varepsilon$ holds for $\req \in \requests$. Let~$\Delta = \sum_{\req \in \requests} (\maxAllocV / \maxDemandV)^2$ and $\beta = 2+\varepsilon \cdot  \sqrt{2\cdot \Delta \cdot \log(|\SV|\cdot|\types|)}$. Then $\mathbb{P} \Big(  A_{\type,u} \geq \beta \cdot \Scap(\type,u) \Big) \leq  (|\SV|\cdot |\types|)^{-4}$ holds, when approximating the cost variant of the VNEP.
\end{lemma}

\begin{restatable}{lemm}{approximationSingleEdgeCost}
\label{lem:approximation-single-edge-cost}
We consider a single edge $(u,v) \in \SE$. Choose $0< \varepsilon \leq 1$, such that $\maxDemandE / \Scap(u,v) \leq \varepsilon $ holds for all $\req \in \requests$. 
Let $\Delta = \sum_{\req \in \requests} (\maxAllocE / \maxDemandE)^2$ and \mbox{$\gamma = 2+\varepsilon\cdot \sqrt{2\cdot \Delta \cdot \log~|\SV|}$}. Then $\mathbb{P} \Big( A_{u,v} \geq \gamma \cdot \Scap(u,v) \Big)  \leq |\SV|^{-4}$ holds, when approximating the cost variant of the VNEP.
\end{restatable}

\noindent\textbf{Main Result.}
Given the Lemmas~\ref{lem:deterministic-objective-guarantee-without}, \ref{lem:approximation-single-node-cost}, and \ref{lem:approximation-single-edge-cost}, a result analogous to Theorem~\ref{thm:result-for-admission-control} can be obtained for the cost variant of the VNEP.

\begin{restatable}{theorem}{mainResultForAdmissionControlCost}
\label{thm:result-for-admission-control-cost}
Assume that $|\SV| \geq 3$. Let $0 < \varepsilon \leq 1$ be chosen minimally, such that  $\maxDemandV[\req][x][y] / \Scap(x,y) \leq \varepsilon$ holds for $(x,y) \in  \SR$ and $\req \in \requests$. Randomized rounding yields a~$(\alpha,\beta,\gamma)$ tri-criteria FPT approximation for the \emph{cost} variant of the VNEP,  such that it finds a solution \emph{with high probabiliy} of cost at most~$\alpha = 2$ times the optimal cost, and cumulative allocations within factors of $\beta$ (nodes) and $\gamma$ (edges) of the original capacities, for $\beta = 2+\varepsilon \cdot \sqrt{2 \cdot \Delta(\SRV) \cdot \log (|\SV|\cdot|\types|)}$ and $\gamma  = 2 +  \varepsilon \cdot \sqrt{2 \cdot \Delta(\SE) \cdot \log |\SV|  }$ with $\Delta:~\mathcal{P}(\SR)~\to~\mathbb{R}_{\geq 0}$ being defined as $\Delta(X) = \max_{(x,y) \in X} \sum_{\req \in \requests} \left(\frac{\maxAllocX}{ \maxDemandX} \right)^2$.
\end{restatable}
}

\section{Extraction Width: Graph Classes and Complexity}
 \label{sec:novel-formulation:size-of-formulation-and-graph-classifications}

  \newcommand{\VC}{\ensuremath{{\textnormal{VC}}}}
  \newcommand{\VCi}[1][i]{\ensuremath{{\textnormal{VC}}_i}}
 
 The runtime of our approximation algorithms grows exponentially in the maximal width of any extraction order. Hence, three questions arise: (i) are there any polynomial-time (unparameterized) approximations for the VNEP, (ii) which graph classes have a bounded extraction width, and (iii) how can extraction orders of minimal width be computed? 
 
The first question can be answered quite easily: there cannot exist polynomial-time approximations (unless $\compPeqNP$), as the following theorem shows.

\begin{theorem}
\label{thm:hardness-of-fractional-vnep}
The \emph{fractional} VNEP is $\NPhard$ and inapproximable (unless $\compPeqNP$). This even holds when request graphs are planar.
\end{theorem}
\begin{proof}
The authors of this paper have recently shown in \cite{rostSchmidVNEPComplexity} that finding \emph{valid} mappings for planar requests is $\NPcomplete$ when allowing for node and edge placement restrictions. Hence, optimizing over a convex combination of \emph{valid mappings} is $\NPhard$ and the fractional VNEP is (polynomial-time) inapproximable (unless $\compPeqNP$). This even holds when not considering substrate graph capacities.
\end{proof}

The above theorem validates our FPT approach to computing solutions to the fractional VNEP, as there cannot exist polynomial-time algorithms for general request graphs and unless $\compPeqNP$ holds. It furthermore underlines that the complexity of the request graphs must be reflected when computing solutions to the fractional VNEP.

 \subsection{Graph Classes of Bounded Extraction Width}
 Given the impossibility of polynomial-time approximations for arbitrary request graphs, we now study graph classes that have bounded extraction width. In particular, we show that `cactus graph requests' and generalizations thereof have bounded extraction width.

 \begin{restatable}{theorem}{lemExtractionWidthCactus}
    \label{lem:extraction-width-cactus}
    Consider a cactus request graph $\VG$, i.e., one for which cycles in its unidrected interpretation intersect in at most a single node. Then $\ewX(\VGbfs) \leq 2$ holds  for \emph{any} extraction order $\VGbfs$. Hence, our approximations run in polynomial-time 
    for cactus graph requests.
    \end{restatable}
    \begin{proof}
    Consider any extraction order $\VGbfs$. $|\VEbfsLabels| \leq 1$ must hold for all edges $e \in \VEbfs$: if this was not the case then two confluences would overlap in $e$ and violate the cactus property. Thus, edge label sets are either equal or disjoint and the maximal edge bag size is~1.
    \end{proof}

While the class of cactus graphs is restrictive, it can be shown that by adding edges parallel to existing ones the width increases by at most the maximum degree of the graph:
\begin{lemma}
 \label{lem:extraction-width-parallel-edges}
 Given an arbitrary graph $\VG$, adding any number of parallel edges (of any direction) for an existing edge does increase the extraction width of $\VG$ by at most the maximum degree of $\VG$. This also holds true if instead paths are added instead of edges.
 \end{lemma}
 \begin{proof}
 Consider an extraction order $\VGbfs$ minimizing the width. Let $e=(i,j) \in \VE$ be an existing edge and assume without loss of generality that the orientation of the edge $e$ is the same in $\VGbfs$. Now, when adding another edge $e'=(i,j)$ or $e''=(j,i)$ to $\VE$, we orient the edge the same way as the original edge $e$. Hence, $e'$ would be introduced to $\VEbfs$ as is, and the orientation of $e''$ would be reversed. Now, if the node $j$ was previously not the target of a confluence, then by introducing $e'$ or $e''$ a new confluence was created and the size of the edge bag of node $i$ containing the edge $(i,j) \in \VEbfs$ increases by one. However, adding $e'$ or $e''$ cannot introduce confluences beyond that: the addition of a \emph{parallel} edge cannot enable a novel confluence to be created.
 Hence, arbitrarily many parallel edges can be added while increasing the size of an edge bag of the tail of the edge by at most one per outgoing edge. Hence, arbitrarily many parallel edges can be created for any existing edge while increasing the sizes of a single edge bag per node by at most one per outgoing edge. Hence, the extraction width of the resulting graph has increased by at most the maximum degree of the original graph $\VG$.
 \end{proof}
 
Lastly, we note that the examples depicted in Figure~\ref{fig:vnet-examples-sc-vc} have small extraction widths.
 
 \begin{observation}
 \label{obs:request-examples}
 The example request graphs depicted in Figure~\ref{fig:vnet-examples-sc-vc} have an extraction width of $2$ and $3$ respectively.
 \end{observation}
 \begin{proof}
The request graph depicted on the right of Figure~\ref{fig:vnet-examples-sc-vc} is a cactus graph and hence has width $2$ by Lemma~\ref{lem:extraction-width-cactus}. 

Considering the request graph on the left, we note that when only considering the solid edges, the graph is a cactus and the width is hence $2$. By adding the `parallel' dashed edges, the extraction width increases by at most the maximum degree $3$ of the cactus graph according to Lemma~\ref{lem:extraction-width-parallel-edges}. However, considering the proof of Lemma~\ref{lem:extraction-width-parallel-edges}, we see that for the node $\textnormal{LB}_1$, only $2$ outgoing edges exist (according to the solid edges) and hence the width increases by at most $2$. Furthermore, as the node $\textnormal{LB}_2$ is already the target of a confluence, the overall extraction width increases by $1$ and the width is hence $3$.
 \end{proof}
 
\subsection{Hardness of Computing Extraction Orders}

Given the above examples, we now study the computational complexity of finding extraction orders of minimum width. In fact, we prove the $\NPhardness$ of computing optimal extraction orders. Our prove relies on a reduction from vertex cover. As a first step towards this goal, consider the following lemma.

  \begin{figure}[b!]

    \centering
    \includegraphics[width=0.9\textwidth]{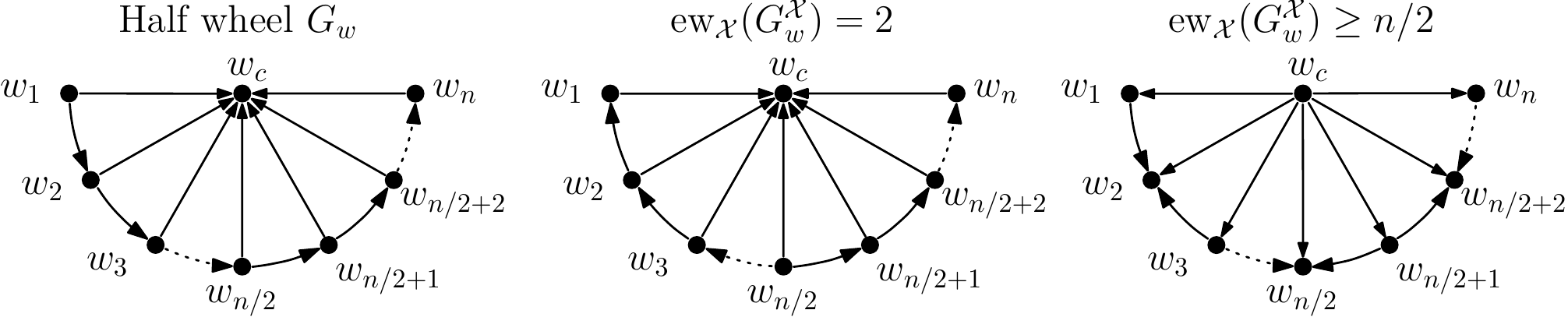}
    \caption{Depicted is an arbitrarily oriented `half wheel' (left) together with two extraction orders: The extraction order in the center is rooted at $w_{n/2}$ and has width $2$. The other order on the right is rooted at $w_c$ and has a width of at least $n/2$ (shown below).}
    \label{fig:half-wheel}
    \vspace{-12pt}
    \end{figure}

\begin{lemma}
\vspace{-4pt}
 \label{lem:vertex-cover-reduction-lem-1}
 Consider the half wheel graph $G_w$ depicted in Figure~\ref{fig:half-wheel} and any extraction order $\VGbfs[w]$ being rooted at $w_c$. Letting $\VC = \{ w_k \in \VV[w]| (w_{k-1},w_k) \in \VEbfs \vee (w_{k+1},w_k) \in \VEbfs\}$, the following holds: $\ewX(\VGbfs[w]) = |\VC| + 1$.
 \end{lemma}
 \begin{proof}
We denote by $G_w^i=(V^i_w,E^i_W)$ the subgraph of $\VGbfs[w]$ induced by the set of nodes $\{w_1, \dots, w_i\} \cup \{w_c\}$. Let $\VCi = \{ w_k \in V^i_w| (w_{k-1},w_k) \in \VEbfs \vee (w_{k+1},w_k) \in \VEbfs\}$. 

Via induction over the subgraphs $G^i_w$ it can be seen that the edges $e_k=(w_c,w_k)$, and $e_{k-1}=(w_c,w_{k-1})$ are either both labeled by $w_k$ (if $w_k \in \VCi$) or by $w_{k-1} \in \VCi$ (if $w_k \notin \VCi$) for all $k \in \{2,\dots,i\}$. 

Observing that $\VC_{n} = \VC$  equals the labels introduced in $\VGbfs[w]$ and noting that the edge label sets $\VEbfsLabels[e_i]$ and $\VEbfsLabels[e_{i+1}]$ overlap for all $i \in \{1,\ldots, n-1\}$, the root $w_c$ must have a single edge bag containing all the labels contained in $|\VCi|$. Hence, $\ewX(\VGbfs[w]) \geq |\VC|+1$ follows. Furthermore, only the nodes contained in $|\VC|$ can be labels (a node not contained in $\VC$ has only a single incoming edge) and the result follows.
 \end{proof}
 
 By Lemma~\ref{lem:vertex-cover-reduction-lem-1}, the following corollary is immediate.
 \begin{corollary}
 \label{cor:minimal-wheel-if-wc-is-root}
Consider a wheel graph $G_w$ with $n$ outer nodes. Considering any extraction order $\VGbfs[w]$ for which the node $w_c$ is chosen to be the root, $\ewX(\VGbfs[w]) \geq \lfloor n/2 \rfloor +1$ holds.
 \end{corollary}
 
 The result of Lemma~\ref{lem:vertex-cover-reduction-lem-1} can be generalized in the following sense.
 
 \begin{lemma}
 \label{lem:vertex-cover-nearly-there}
 Given is a connected, undirected graph $\bar{G}=(\bar{V},\bar{E})$, we define a directed version ${G}_{VC}=({V}_{VC},{E}_{VC})$ with an additional super node $\hat{r}$ as follows: ${V}_{VC} = \bar{V} \sqcup \{\hat{r}\}$, and ${E}_{VC} = \{(i,j) | \{i,j\} \in \bar{E}, i < j \} \cup \{(\hat{r},i)| i \in \bar{V}\}$. 
 The minimal width of an extraction order ${G}^{\mathcal{X}}_{VC}$ rooted at $\hat{r}$ equals the size of the minimum vertex cover  of $\bar{G}$  plus one.
 \end{lemma}
 \begin{proof}
 Let ${G}^{\mathcal{X}}_{VC}$ be an extraction order of $G_{VC}$ which is rooted at $\hat{r}$.
 The proof of Lemma~\ref{lem:vertex-cover-reduction-lem-1} has shown that whenever a path $P$ in the original graph is considered, all nodes of ${G}^{\mathcal{X}}_{VC}$ with at least one incoming edge (with respect to the original edge set) are labels of the \emph{same} edge bag of the root $\hat{r}$. As this property holds for any simple path contained in $\bar{G}$ and as $\bar{G}$ is connected, there can only be a single edge bag: if there was more than one edge bag, then there does not exist a path $P$ connecting any of the edges of the first bag to any of the edges in the second bag, refuting the connectivity of $\bar{G}$.  Applying Lemma~\ref{lem:vertex-cover-reduction-lem-1} for any path $P$ of the original graph, the single edge bag of the root $\hat{r}$ must contain any node having at least one incoming edge according to the original edge set $\bar{E}$. Hence, assuming that $G^{\mathcal{X}}_{VC}$ has minimal width, the width of ${G}^{\mathcal{X}}_{VC}$ equals the size of the minimum vertex cover of $\bar{G}$ plus one.
 \end{proof}
 
 Lemma~\ref{lem:vertex-cover-nearly-there} is the basis of our proof that computing extraction orders of minimal width is $\NPhard$ via a reduction from vertex cover (cf. Theorem~\ref{thm:finding-extraction-order-np-hard}). For the proof of our reduction, we require the following lemma.
 
{
\renewcommand{\VG}{\ensuremath{G}}
\renewcommand{\VV}{\ensuremath{V}}
\renewcommand{\VE}{\ensuremath{E}}
\renewcommand{\VGbfs}{\ensuremath{G^{\mathcal{X}}}}
\renewcommand{\VEbfs}{\ensuremath{E^{\mathcal{X}}}}
\renewcommand{\VVroot}{\ensuremath{s}}

\begin{lemma}
\label{lem:separator-becomes-root}
Consider a graph $\VG=(\VV,\VE)$ with a corresponding extraction order $\VGbfs=(\VV,\VEbfs, \VVroot)$. Assume that a node $v \in \VV$ exists that separates a set of nodes $U \subset \VV$  from the root node $\VVroot$. Then any edge incident to $v$ and some node $u \in U$ is oriented away from $v$ in the extraction order $\VGbfs$, i.e. $(v,u) \in \VEbfs$ holds for all $u \in U$.
\end{lemma}
\begin{proof}
Assume for the sake of contradiction that for some node $u \in U$ the edge $(u,v)$ is contained in the extraction order. By the definition of the extraction order all nodes must be reachable from the root $\VVroot$. As the node $v$ separates $u$ from the root, all paths from $\VVroot$ to $u$ must contain $v$. Hence, $u$ must be reachable from $v$ and the edge $(u,v)$ hence creates a loop in $\VGbfs$ which contradicts the acyclicity of $\VGbfs$. Hence, any edge incident to $v$ and some node $u \in U$ must be directed away from $v$.
\end{proof}
}

 {
 \renewcommand{\VGbfs}{\ensuremath{G_{VC}^{\mathcal{X}}}}
 \renewcommand{\VEbfs}{\ensuremath{E_{VC}^{\mathcal{X}}}}
 \renewcommand{\VVroot}{\ensuremath{s_{VC}}}

\begin{restatable}{theorem}{thmFindingExtractionOrderNPHard}
 \label{thm:finding-extraction-order-np-hard}
 Computing an extraction order of minimum width is $\NPhard$.
 \end{restatable}
 \begin{proof}
 We give a polynomial time reduction of the vertex cover problem to the problem of finding the extraction order of minimum width. We adapt the construction used in Lemma~\ref{lem:vertex-cover-nearly-there} slightly, to force the mapping of the root node to $\hat{r}$. Concretely, we add a \emph{half wheel graph (cf. Figure~\ref{fig:half-wheel})} $G_w$ with $2\cdot|\bar{V}|+2$ outer nodes to the graph ${G}_{VC}$ \emph{and identify} the node $\hat{r}$ with the wheel's node $w_c$, i.e. $\hat{r} = w_c$. Let $\VGbfs = (V_{VC}, \VEbfs, \VVroot)$ be an extraction order of minimum width. The extraction order's root $\VVroot$ must be placed on some outer wheel node:
 \begin{description}
 \item[Root $\VVroot$ is placed on a wheel node $w_i$:] We first consider the orientations of edges inside the wheel graph. According to  Figure~\ref{fig:half-wheel} (center) there is an orientation such that the extraction width inside the wheel graph is 2. The node $\hat{r}=w_c$ separates the original graph $\bar{G}$ from the extraction order's root $\VVroot=w_i$. Thus, by Lemma~\ref{lem:separator-becomes-root} all edges incident to a node  $v \in \bar{V}$ and $w_c$ must be oriented away from $w_c$. Hence, excluding the outer wheel nodes, the node $w_c$ is a root in the corresponding extraction order. Thus, the width of the extraction order $\VGbfs$ -- excluding the outer wheel nodes -- equals the size of a minimum vertex cover of $\bar{G}$ plus one by Lemma~\ref{lem:vertex-cover-nearly-there} and the assumption that $\VGbfs$ is of minimal width. Lastly, note that no confluence spanning the wheel graph $G_w$ and the graph $\bar{G}$ exists. Letting $VC$ denote a minimal vertex cover of $\bar{G}$, the width of the extraction order $\VGbfs$ equals $\max\{2,|VC|\} \leq |\bar{V}|$.
 \item[Root $\VVroot$ is placed on $\hat{r}=w_c$:] In this case, the width of the extraction order ${G}^{\mathcal{X}}_{VC}$ is at least $|\bar{V}|+1$ based on Corollary~\ref{cor:minimal-wheel-if-wc-is-root}. Hence, as the size of a vertex cover of $\bar{G}$ is always less than $|\bar{V}|$, the placement of the root on $\hat{r}$ contradicts the optimality assumption of $\VGbfs$.
 \item[Root $\VVroot$ is placed on a node $v \in \bar{V}$:] In this case, the width is again at least $|\bar{V}|+1$: as the node $\hat{r}=w_c$ separates the outer wheel nodes from the root $\VVroot$, all wheel edges incident to $w_c$ are oriented away from $w_c$. As $w_c$ is hence a root in the extraction order restricted to the wheel graph, the width is at least $|\bar{V}|+1$ by Lemma~\ref{lem:vertex-cover-reduction-lem-1}. Thus, the placement of the extraction's root on any node $v \in \bar{V}$ contradicts the optimality of the extraction order $\VGbfs$.
 \end{description}
 Now, let $VC\subseteq \bar{V}$ denote a minimum vertex cover of $\bar{G}$. If $|VC| > 1$ holds, then for any optimal extraction order $\VGbfs$, $\ewX(\VGbfs) = |VC|+1$ holds. Furthermore, a minimal vertex cover $VC$ can be recovered from any minimum width extraction order $\VGbfs$ by placing any node $v$ in the cover $VC$ whenever at least two edges are oriented towards it in $\VGbfs$. As the cases in which the minimal vertex cover is less or equal to $1$ can be trivially identified, computing a minimum width extraction order is $\NPhard$.
 \end{proof}
 }

\section{Conclusion}\label{sec:conclusion}
We have presented the first (fixed-parameter tractable) approximation algorithms
for the Virtual Network Embedding Problem (VNEP)
supporting arbitrary request graphs.
To enable the decomposability of general request graphs, we have developed a novel LP formulation whose size is parameterized by a novel graph number: the extraction width and exploring it further will be of great interest for practical applications. Finally, while having focused on the theoretic aspects of approximating the VNEP, we provide the research community with implementations and empirical evaluations 
at \url{https://vnep-approx.github.io/}.

\noindent\textbf{Acknowledgements.} This work was partially supported by Aalborg University's PreLytics project as well as by the German BMBF Software Campus grant 01IS1205.
The authors would like to thank Elias D\"ohne and Alexander Elvers for proof-reading parts of the paper and contributing significantly to our implementation at  \url{https://vnep-approx.github.io/}.

{
\bibliographystyle{plainurl}

}

\appendix

\vspace{-8pt}
\section{Deferred Proofs}
\vspace{-6pt}
\label{app:deferred-proofs}

\approximationSingleEdge*
\begin{proof}
\vspace{-4pt}
Each random variable~$A_{r,u,v}$ is clearly contained in the interval~$[0, \maxAllocE]$. We choose~$t = \varepsilon \cdot \sqrt{2 \cdot \Delta \cdot \log |\SV|   } \cdot \Scap(u,v)$ and apply Hoeffding's Inequality:
{
\vspace{-4pt}
\allowdisplaybreaks
\begin{alignat}{2}
& \mathbb{P} \Big(A_{u,v} - \mathbb{E}(A_{u,v})~\geq t \Big)~&&  \notag \\
& ~ \leq exp \Bigg(\frac{-2\cdot t^2}{\sum_{\req \in \requests}~(\maxAllocE)^2} \Bigg)~&& \\
& ~ \leq exp \Bigg(\frac{-4 \cdot \varepsilon^2 \cdot \Delta \cdot \log |\SV| \cdot  \Scap^2(u,v)}{\sum_{\req \in \requests} (\varepsilon \cdot \Scap(\type,u)~\cdot \maxAllocE / \maxDemandE)^2} \Bigg)~&& \label{formula:proof:singleEdgeViolation:upperCapacities}\\
& ~ \leq exp \Bigg(\frac{-4 \cdot \varepsilon^2 \cdot \sum_{\req \in \requests} (\maxAllocE / \maxDemandE)^2 \cdot \log |\SV| \cdot \Scap^2(u,v)}{\varepsilon^2\cdot\Scap^2(u,v) \cdot \sum_{\req \in \requests} (\maxAllocE / \maxDemandE)^2} \Bigg)= |\SV|^{-4}&& \label{formula:proof:singleEdgeViolation:upperCapacities_2} 
\end{alignat}

}
\vspace{-4pt}
We have used $\maxAllocV \leq \varepsilon \cdot \Scap(\type,u)~\cdot \maxAllocV / \maxDemandV$ in the second step, which follows from  $\maxDemandV \leq \varepsilon \cdot \Scap(\type,u)$. In the third step we plugged in the definition of $\Delta$. Lastly, we utilize that the expected load $\mathbb{E}(A_{\type,u})$ is upper bounded by the resource's capacity $\Scap(\tau,u)$ to obtain the lemma's statement.
\end{proof}

\end{document}